\documentclass{article}
\usepackage[T2A]{fontenc}
\usepackage[cp1251]{inputenc}
\usepackage[english,russian]{babel}
\usepackage[tbtags]{amsmath}
\usepackage{amsfonts,amssymb,mathrsfs,amscd}
\usepackage{array}
\usepackage[hyper]{msbmy}
\JournalName{}
\numberwithin{equation}{section}
\def\VOLUME{203}
\def\ISSUE{2, 111--142}
\def\YEAR{2012}
\newcounter{myta}
\numberwithin{myta}{section}
\newcommand{\myt}{\refstepcounter{myta}\themyta}
\theoremstyle{plain}
\newtheorem{theorem}{Теорема}

\newtheorem{propos}{Предложение}
\theoremstyle{definition}
\newtheorem{definition}{Определение}
\newtheorem{proof}{Доказательство}
\newtheorem{remark}{Замечание}

\def\eps{\varepsilon}

\def\rk{\operatorname{rank}}
\def\rank{\operatorname{rank}}

\def\sgrad{\operatorname{sgrad}}
\def\Ker{\operatorname{Ker}}
\def\Im{\operatorname{Im}}
\def\sp{\operatorname{sp}}
\def\bR{\mathbb{R}}
\def\bZ{\mathbb{Z}}
\def\bs{\boldsymbol}
\newcommand {\ri} {\mathrm{i}}
\newcommand {\fts}[1] {{\small #1}}
\newcommand {\mstrut}{\vphantom{\bigl(}}
\newcommand {\mF}{\mathcal{F}}

\newcommand {\ds}{\displaystyle}
\newcommand {\mbf}[1] {\mathbf{#1}}
\newcommand {\mA} {\mathcal{A}}
\newcommand {\mP} {\mathcal{P}}
\newcommand {\mK} {\mathcal{C}}
\newcommand {\mM} {\mathcal{M}}
\newcommand {\mL} {\mathcal{L}}
\newcommand {\LS} {\mathcal{S}}
\newcommand {\mm} {\mathcal{M}_1}
\newcommand {\mn} {\mathcal{M}_2}
\newcommand {\mo} {\mathcal{M}_3}
\newcommand {\ml} {\mathcal{M}_4}

\newcommand {\pov} {\Gamma}
\newcommand {\ad}[1] {\mathop{\rm Adiag}\nolimits \{#1\}}
\begin{document}

\title{Классификация особенностей в задаче о движении волчка Ковалевской в двойном поле
сил}
\author[P.\,E.~Ryabov]{П.\,Е.~Рябов}
\address{Финансовый университет при Правительстве Российской Федерации}
\email{orelryabov@mail.ru}
\author[M.\,P.~Kharlamov]{М.\,П.~Харламов}
\address{Волгоградская академия государственной службы}
\email{mharlamov@vags.ru}

\date{22.06.2010}
\udk{517.938.5; 531.38}

\maketitle

\begin{fulltext}

\begin{abstract}
Рассматривается задача о движении волчка
Ковалевской в двойном силовом поле (случай
интегрируемости А.Г.\,Реймана --
М.А.\,Семенова-Тян\-Шанского без
гиростатического момента). Это вполне
интегрируемая гамильтонова система с тремя
степенями свободы, несводимая к семейству
систем с двумя степенями свободы. Изучается
критическое множество интегрального
отображения. Приводится описание критических
подсистем и бифуркационных диаграмм. Дана
классификация всех невырожденных критических
точек --- положений равновесия (невырожденных
особенностей ранга $0$), особых периодических
движений (невырожденных особенностей ранга
$1$), а также критических двухчастотных
движений (невырожденных особенностей ранга
$2$).

Библиография: 32 назв.
\end{abstract}


\markright{Особенности волчка Ковалевской в двойном поле сил}

\footnotetext[0]{Работа выполнена при поддержке РФФИ (грант
\No~10-01-00043).}

\footnotetext[0]{\textbf{Опубликовано:} Математический сборник, 2012, 203(2), 111--142}

\tableofcontents

\section{Введение}\label{s1}
Задача о движении волчка Ковалевской в двойном поле сил описывается
системой уравнений \cite{Bogo}
\begin{equation}\label{eq1.1}
\begin{array}{lll}
\displaystyle{2 \, \dot \omega_1=\omega_2\omega_3+\beta_3,}
&\displaystyle{ \dot \alpha_1=\alpha_2\omega_3-\alpha_3\omega_2,} &
\displaystyle{\dot\beta_1=\beta_2\omega_3-\omega_2\beta_3,}\\
\displaystyle{2 \, \dot \omega_2 = - \omega_1 \omega_3-\alpha_3,} &
\displaystyle{\dot \alpha_2=\omega_1\alpha_3-\omega_3\alpha_1,}&
\displaystyle{\dot\beta_2=\omega_1\beta_3-\omega_3\beta_1,}\\
\displaystyle{\dot\omega_3=\alpha_2-\beta_1,}& \displaystyle{\dot
\alpha_3=\alpha_1\omega_2-\alpha_2\omega_1,} & \displaystyle{\dot
\beta_3=\beta_1\omega_2-\beta_2\omega_1.}
\end{array}
\end{equation}
Здесь $\boldsymbol\omega$ -- вектор мгновенной угловой скорости.
Постоянные в инерциальном пространстве векторы $\boldsymbol \alpha,
\boldsymbol\beta$ характеризуют действие силовых полей. Обозначения
выберем так, чтобы выполнялось неравенство
$|\boldsymbol\alpha|\geqslant|\boldsymbol\beta|.$ Как показано в
\cite{Kh34}, без ограничения общности силовые поля можно считать
взаимно ортогональными. Тогда геометрические интегралы системы
(\ref{eq1.1}) запишутся в виде ($a \geqslant b \geqslant 0$)
\begin{eqnarray}\label{eq1.4}
|{\boldsymbol\alpha}|^2=a^2,\quad |{\boldsymbol\beta}|^2=b^2,\quad
{\boldsymbol\alpha}\cdot {\boldsymbol\beta}=0.
\end{eqnarray}

Используя компоненты кинетического момента
\begin{eqnarray*}\label{eq1.3}
M_1 =2 \omega_1, \quad M_2 =2 \omega_2, \quad M_3 = \omega_3,
\end{eqnarray*}
перенесем на пространство ${\bR}^9({\boldsymbol
\omega,\boldsymbol\alpha, \boldsymbol\beta})$ введенные в работе
\cite{Bogo} скобки Пуассона
\begin{eqnarray}\label{eq1.2}
\begin{array}{c}
\{M_i,M_j\}=\eps_{i j k} M_k, \quad \{M_i,\alpha_j\}=\eps_{i j
k} \alpha_k, \quad \{M_i,\beta_j\}=\eps _ {i j k} \beta_k, \\[1.5mm]
\{\alpha_i,\alpha_j\}=0, \quad \{\alpha_i,\beta_j\}=0, \quad
\{\beta_i,\beta_j\}=0.
\end{array}
\end{eqnarray}
Система (\ref{eq1.1}) примет вид $\dot x=\{x, H\}$, где $x$ -- любая
из координат, а
\begin{eqnarray*}\label{equa:H}
H=\frac{1}{2}(2\omega_1^2+2\omega_2^2+\omega_3^2)-\alpha_1-\beta_2.
\end{eqnarray*}
Функциями Казимира для скобок (\ref{eq1.2}) являются левые части
уравнений (\ref{eq1.4}). Поэтому векторное поле (\ref{eq1.1}),
ограниченное на заданное этими уравнениями шестимерное подмногообразие
$\mP^6$ в ${\bR}^9(\boldsymbol\omega,\boldsymbol\alpha,\boldsymbol\beta)$, является гамильтоновой системой с тремя
степенями свободы.

При $b=0$ система (\ref{eq1.1}) описывает случай С.В.\,Ковалевской
движения твердого тела в поле силы тяжести, а при $a=b$ --- cлучай Х.М.\,Яхья
\cite{Yeh}. Для этих предельных случаев задачи обладают
группой симметрий и редуцируются к семействам
интегрируемых систем с двумя степенями свободы (конфигурационное пространство -- сфера).
Классический случай Ковалевской изучен в
\cite{KhPMM}, \cite{KhDAN}, \cite{BolFomRich}.
Случай Яхья и его обобщения рассматривались в \cite{JPA}.
Глобальный подход к классификации интегрируемых систем на двумерной сфере, порожденных задачами динамики твердого тела с осесимметричным потенциалом, реализован в работах \cite{Fom1991},
\cite{BolFom1994}, \cite{BolKozFom1995}.

Далее предполагается, что
\begin{equation}\label{aneb}
a>b>0.
\end{equation}
Функции
\begin{equation*}\label{equa:GK}
\begin{array}{l}
K=(\omega_1^2-\omega_2^2+\alpha_1-\beta_2)^2+(2\omega_1\omega_2+\alpha_2+\beta_1)^2,\\
G=\left[\omega_1\alpha_1+\omega_2\alpha_2+\frac{1}{2} \alpha_3
\omega_3\right]^2 +
\left[ \omega_1 \beta_1 + \omega_2 \beta_2 + \frac{1}{2} \beta_3 \omega_3 \right]^2+\\
\phantom{G=}+\omega_3\left[(\alpha_2\beta_3-\alpha_3\beta_2)\omega_1+(\alpha_3\beta_1-\alpha_1\beta_3)\omega_2+
\frac{1}{2}(\alpha_1\beta_2-\alpha_2\beta_1)\omega_3\right]-\\
\phantom{K=}-\alpha_1b^2-\beta_2 a^2
\end{array}
\end{equation*}
вместе с $H$ образуют на $\mP^6$ полный
инволютивный набор интегралов системы
(\ref{eq1.1}). Интегралы $K$ и $G$ указаны в
\cite{Bogo} и \cite{BobReySem}.

Определим интегральное отображение
\begin{equation}\label{eqmom}
{\mF }: \mP^{6} \to {\bR}^3,
\end{equation}
полагая ${\mF }(x)=\bigl( g=G(x), k=K(x), h=H(x)\bigr)$. Отображение
${\mF }$ принято называть {\it отображением момента\/}.

Обозначим через $\mK$ совокупность всех критических точек
отображения момента, то есть точек, в которых $\rk d{\mF }(x)<3$.
Множество критических значений $\Sigma={\mF }(\mK) \subset{\bR}^3$
называется {\it бифуркационной диаграммой}. Множество $\mK$ можно
стратифицировать рангом отображения момента, представив в виде
объединения $\mK =\mK^0 \cup \mK^1 \cup \mK^2$. Здесь $ \mK^r= \{x
\in \mP^6 | \rank d{\mF}(x)=r \}$. В соответствии с этим и диаграмма
$\Sigma$ становится клеточным комплексом $\Sigma =\Sigma^0 \cup
\Sigma^1 \cup \Sigma^2$. С другой стороны, на практике
бифуркационные диаграммы описываются в терминах некоторых
поверхностей в пространстве констант первых интегралов. Уравнения
этих поверхностей (неявные или параметрические) зачастую можно
получить даже не вычисляя самих критических точек, как
дискриминантные множества некоторых многочленов (например, исходя из
особенностей алгебраических кривых, ассоциированных с
представлениями Лакса). Такие поверхности будем обозначать через
$\pov _i$ и записывать представление
\begin{equation*}\label{strsig}
\Sigma= \bigcup_i \Sigma_i, \qquad \Sigma_i = \Sigma \cap \pov _i.
\end{equation*}
При этом, в свою очередь, каждое множество $\Sigma_i$
стратифицировано
\begin{equation*}\label{strsigi}
\Sigma_i = \Sigma_i^0 \cup \Sigma_i^1 \cup \Sigma_i^2
\end{equation*}
и различные множества $\Sigma_i$ могут пересекаться между собой по
стратам размерности меньше~2. Смысл такого представления в том, что
критическое множество $\mK$ оказывается объединением естественным
образом возникающих инвариантных множеств $\mM_i=\mK \cap
\mF^{-1}(\pov_i)$. Если поверхность $\pov_i$ записана регулярным
уравнением
\begin{equation}\label{eqpov}
\phi_i(g,k,h) = 0,
\end{equation}
то $\mM_i$ определится как множество критических точек интеграла
$\phi_i(G,K,H)$, лежащих на его нулевом уровне, а вычисленные в
точке $\mM_i$ компоненты градиента функции $\phi_i$ в подстановке
значений интегралов $G,K,H$ дадут коэффициенты равной нулю линейной
комбинации дифференциалов $dG,dK,dH$. В точке трансверсального
пересечения двух поверхностей $\pov_i$ и $\pov_j$ получим две
независимые равные нулю комбинации, поэтому в точках
соответствующего пересечения $\mM_i \cap \mM_j$ ранг $\mF$ равен~1.
Очевидно, что точки трансверсального пересечения трех поверхностей
(углы бифуркационной диаграммы) оказываются порожденными точками с
условием $\rank \mF =0$. Множества $\mM_i$ с индуцированной на них
динамикой далее называем критическими подсистемами.

Критические подсистемы и уравнения поверхностей $\pov_i$ в
рассматриваемой задаче найдены в работе \cite{Kh2005}. Подробное
описание стратификации критического множества по рангу отображения
момента изложено в \cite{Kh36}. Там же в виде явных неравенств
указаны области существования движений на поверхностях $\pov_i$ --
множества $\Sigma_i$, составляющие бифуркационную диаграмму. Как
следствие построен атлас всех сечений диаграммы $\Sigma$ плоскостями
постоянной энергии, то есть найдены все бифуркационные диаграммы
$\Sigma(h)$ отображения $G{\times}K$, ограниченного на
изоэнергетические поверхности $\{H=h\} \subset \mP^6$. Краткая
сводка этих результатов в необходимом объеме приведена ниже.

Критические подсистемы оказываются интегрируемыми почти всюду гамильтоновыми системами с числом степеней свободы меньшим трех. Для них, в свою очередь, определено индуцированное отображение момента. Бифуркационная диаграмма $\Sigma_i^*$ для отображения $\mF|_{\mM_i}$ естественным образом отождествляется с объединением $\Sigma_i^0 \cup \Sigma_i^1$. Описание критических множеств, диаграмм $\Sigma_i^*$ и бифуркаций {\it внутри} критических подсистем получено в работах \cite{Zot}, \cite{KhSav}, \cite{Kh2009}. Однако классификация точек множества $\mK$ по отношению ко всей системе с тремя степенями свободы на $\mP^6$ не проводилась. В данной работе исследуется тип невырожденных особенностей полного отображения момента.

\section{Понятие невырожденной особенности}\label{s2}
Напомним некоторые определения и факты, связанные с особенностями отображения момента и бифуркациями в случае многих степеней свободы \cite{Fom1986}, \cite{Fom1989},
\cite{BolFom}, \cite{BolsOsh2006}. Пусть
\begin{equation}
\label{eq2.1} (\mP^{2k},\Omega, H, {\mF } )
\end{equation}
-- интегрируемая по Лиувиллю гамильтонова система. Здесь
$(\mP^{2k},\Omega)$ -- симплектическое многообразие, динамика задана
полем $v=\sgrad H$, и $f_1,\ldots,f_k$ -- независимые (почти всюду)
интегралы в инволюции, составляющие отображение момента
\begin{equation*}
{\mF }: \mP^{2k} \to {\bR}^k, \qquad {\mF}(x)=(f_1(x),\ldots,f_k(x))
\end{equation*}
с множеством критических точек $\mK =\mK^0 \cup \mK^1 \cup \ldots
\cup \mK^{k-1}$, где
$$
\mK^r=\{x\in \mP^{2k}| \rank d{\mF }(x)=r\}.
$$

Для любой точки $x\in\mK^r$  можно подобрать невырожденную линейную
замену функций $(f_1,\ldots,f_k) \mapsto (g_1,\ldots,g_k)$, такую,
что

1) ${dg_1(x)=\ldots=dg_{k-r}(x)=0}$;

2) дифференциалы $dg_{k-r+1}(x),\ldots,dg_k(x)$ линейно независимы.

Пусть $\varphi$ -- гладкая функция на $\mP^{2k}$, для которой точка
$x$ критическая. Линеаризация $A_\varphi(x)$ векторного поля $\sgrad
\varphi$ в точке $x$ является симплектическим оператором в
касательном пространстве $T_x\mP^{2k}$. Алгебру таких операторов
отождествим с $\sp(2k,{\bR})$. Там, где это не приведет к
недоразумению, обозначение точки в операторе вида $A_\varphi$
опускаем.

Пусть $x\in \mK^r$, $j \in \{1,\ldots,k-r\}$. Обозначим $v_j=\sgrad
g_j$. Рассмотрим в $T_x\mP^{2k}$ подпространство $W$, натянутое на
$v_{k-r+1}$,\ldots, $v_k$, и его косоортогональное дополнение
$W^\prime$. Тогда ${W \subset \Ker A_{g_j}}$ и ${\Im A_{g_j} \subset
W^\prime}$. На факторпространстве $W^\prime/W$ имеется
симплектическая структура $\tilde\Omega$, а операторы $\tilde
A_{g_j}=A_{g_j}|_{W^\prime/W}$ являются элементами алгебры Ли
${\sp(2(k-r),{\bR})}$. Таким образом, в вещественной симплектической
алгебре Ли ${\sp(2(k-r),{\bR})}$ можно рассмотреть коммутативную
подалгебру $\mA(x,{\mF })$, порожденную операторами $\tilde
A_{g_1},\ldots,\tilde A_{g_{k-r}}$.

\begin{definition}
Критическая точка $x\in \mK^r$ отображения момента $\mF $ называется
\textit{невырожденной ранга} $r$ (\textit{коранга} $k-r$), если
$\mA(x,{\mF })$ -- подалгебра Картана в $\sp(2(k-r),{\bR})$.
\end{definition}

Коммутативная подалгебра в $\sp(2(k-r),{\bR})$ является картановской
тогда и только тогда, когда ее размерность равна $k-r$ и среди ее
элементов найдется линейный оператор с различными собственными
значениями (он называется регулярным оператором). Заметим, что
касательное пространство $T_{x}\mP^{2k}$ раскладывается на два
подпространства: $T_{x}\mP^{2k}=V_1\oplus V_2$, где $V_1$ --
корневое подпространство, которое соответствует собственному
значению $0$ алгебраической кратности $2r$, а $V_2$ изоморфно
$W^\prime/W$. Характеристическое уравнение регулярного оператора
$A_{g_j}$ (если такой найдется), ограниченного на подпространство
$V_2 \cong W^\prime/W$, имеет такой же набор $2(k-r)$ различных
собственных значений, что и регулярный элемент из картановской
подалгебры $\mA(x,{\mF })$ в $\sp(2(k-r),{\bR})$.

Для невырожденных критических точек определяется модельное слоение
Лиувилля: в пространстве ${\bR}^{2k}$ с координатами
$q_1,\ldots,q_k$, $p_1,\ldots,p_k$ и стандартной симплектической
структурой задается набор функций
\begin{equation}\label{pfuns}
\begin{array}{l}
P_i^{ell}=p_i^2+q_i^2 \quad (i=1,\ldots,m_1);\\
P_j^{hyp}=p_j q_j \quad (j=m_1+1,\ldots, m_1+m_2);\\
P_l^{foc}=p_lq_l+p_{l+1}q_{l+1}, \quad P_{l+1}^{foc}=p_lq_{l+1}-p_{l+1}q_l\\
\quad (l=m_1+m_2+1,m_1+m_2+3,\ldots,m_1+m_2+2m_3-1);\\
P_s=q_s \quad  (s=m_1+m_2+2m_3+1,\ldots,k).
\end{array}
\end{equation}
Их совместные уровни
определяют \textit{каноническое лиувиллево
слоение ${\LS}_{can}$} в ${\bR}^{2k}$  в
окрестности точки $0$, которая является
невырожденной особой точкой отображения момента
${\mF }_{can}:{\bR}^{2k}\to{\bR}^{k}$,
порожденного функциями (\ref{pfuns}). Ранг этой
невырожденной особенности равен $r$, а тройка
целых чисел $(m_1,m_2,m_3)$ называется ее типом
\cite{BolFom}. Таким образом, любая
невырожденная особенность отображения момента
${\mF }_{can}$ кратко характеризуется четверкой
целых чисел $(r, m_1,m_2,m_3)$, где $r$ -- ранг
отображения в особой точке, $(m_1,m_2,m_3)$ --
тип точки.  При этом $m_1+m_2+2m_3+r=k$, где
$k$ -- число степеней свободы рассматриваемой
системы.

Следующая теорема \cite{BolFom} позволяет понять, как выглядит
бифуркационная диаграмма в окрестности невырожденной особой точки.
Она будет диффеоморфна бифуркационной диаграмме ${\Sigma}_{can}$
модельного отображения момента ${\mF }_{can}$.

\begin{theorem}
Пусть дана вещественно-аналитическая интегрируемая система с $k$
степенями свободы на вещественно-аналитическом многообразии
$\mP^{2k}$. Тогда слоение Лиувилля в окрестности невырожденной
особой точки ранга $r$ и типа $(m_1,m_2,m_3)$ всегда локально
симплектоморфно модельному слоению Лиувилля ${\LS}_{can}$ с теми же
параметрами. В частности, два слоения Лиувилля с совпадающими
параметрами $(r,m_1,m_2,m_3)$ локально симплектоморфны.
\end{theorem}

Как показано в работах \cite{MirandaZung2004}, \cite{Zung2006},
утверждение теоремы остается верным, если невырожденную особую точку
в формулировке теоремы заменить на \textit{компактную невырожденную
особую орбиту} пуассоновского действия, ассоциированного с системой
(\ref{eq2.1}). Такие орбиты могут быть организованы в критические
подмногообразия.
\begin{definition}
Множество ${\cal M}\subset \mP^{2k}$ называется {\it критическим
подмногообразием} ранга $m<k$ (коранга $k-m$), если

$1)$ ${\cal M}$ -- гладкое подмногообразие
в $\mP^{2k}$ размерности $2m$;

$2)$ $\Omega_{\cal M}=\Omega\bigr|_{\cal
M}$ -- есть симплектическая структура на ${\cal M}$;

$3)$ ${\cal M}$ инвариантно относительно
векторного поля $v=\sgrad H$;

$4)$ $\rk {\mF }=m$ почти всюду на ${\cal M}$.
\end{definition}
\begin{propos}
Если ${\cal M}$ -- критическое подмногообразие ранга $m$, то
$$
({\cal M},\Omega\bigr|_{\cal M}, H\bigr|_{\cal M})
$$
-- вполне интегрируемая гамильтонова система с $m$ степенями свободы.
\end{propos}

Описанные в предыдущем разделе критические
подсистемы оказываются критическим
многообразиями или их замыканиями. В последнем
случае они являются критическими многообразиями
почти всюду. Кроме нарушений гладкости фазовых
пространств критических подсистем могут
возникать особенности, связанные с вырождением
индуцированной симплектической структуры. Это
явление было впервые обнаружено Д.Б.\,Зотьевым
\cite{Zot}.

Следующее утверждение позволяет установить глобальный тип
критической точки, принадлежащей двум критическим подсистемам, зная
ее типы относительно каждой. Пусть ${\cal M}$ и ${\cal N}$ --
критические подмногообразия в $\mP^{2k}$ ранга $m=k-1$, а ${\cal
M}\cap{\cal N}={\cal L}$ -- критическое подмногообразие ранга
$r=k-2$ и пусть пересечение ${\cal M}\cap{\cal N}$ трансверсально.

\begin{propos}\label{propm}
Пусть $x_0\in {\cal L}$ -- невырожденная критическая точка системы
$(\ref{eq2.1})$ ранга $k-2$, имеющая тип ${(\tau_1,\tau_2,0)}$, где
$\tau_1+\tau_2=2$. Тогда

$1)$ для подсистемы на ${\cal L}$ эта точка регулярна;

$2)$ для подсистем на ${\cal M}$ и на ${\cal N}$ эта точка критическая коранга $1$, причем в подсистемах
${\cal M}$ и ${\cal N}$ ее тип имеет вид ${(m_1,m_2,0)}$ и ${(n_1,n_2,0)}$ соответственно, а числа $m_i, n_i$ связаны
соотношениями $\tau_1=m_1+n_1$, $\tau_2=m_2+n_2$, $m_1+m_2=1$,$n_1+n_2=1$.
\end{propos}
\begin{proof} По предположению имеем
\begin{eqnarray*}
&& T_{x_0}\mP^{2k}=T_{x_0}{\cal M}+T_{x_0}{\cal N},\quad
T_{x_0}{\cal M}\cap T_{x_0}{\cal N}=T_{x_0}{\cal L},\\
&& \dim T_{x_0}{\cal M}=\dim T_{x_0}{\cal N}=2(k-1),\quad \dim
T_{x_0}{\cal L}=2(k-2),
\end{eqnarray*}
и все указанные пространства инвариантны для операторов из
$\sp(2(k-2),{\bR})$, порождающих соответствующую подалгебру
Картана.\end{proof}

 Перейдем к описанию критических точек
обобщенного волчка Ковалевской.

\section{Критические подсистемы и бифуркационные диаграммы}\label{s3}
В этом разделе излагается сводка результатов
\cite{Kh2005,Kh36, Kh2004, Kh2006}, относящихся
к нахождению критического множества отображения
момента (\ref{eqmom}) и бифуркационных
диаграмм. Приводится система единых
обозначений, необходимая для трехмерной
классификации.

Для описания критических подсистем удобно воспользоваться заменой
переменных, введенной в работе \cite{Kh32}, обобщающей замену
С.В.\,Ковалевской и подсказанной представлением Лакса
\cite{BobReySem}:
\begin{equation*}\label{eq1.5}
\begin{array}{l}
\begin{array}{ll}
x_1 = (\alpha_1  - \beta_2) + \ri (\alpha_2  + \beta_1),&
x_2 = (\alpha_1  - \beta_2) - \ri (\alpha_2  + \beta_1 ), \\
y_1 = (\alpha_1  + \beta_2) + \ri (\alpha_2  - \beta_1), & y_2 =
(\alpha_1  + \beta_2) -
\ri (\alpha_2  - \beta_1), \\
 z_1 = \alpha_3  + \ri \beta_3, &
z_2 = \alpha_3  - \ri \beta_3,
\end{array}\\
\begin{array}{lll}
w_1 = \omega_1  + \ri \omega_2 , & w_2 = \omega_1  - \ri \omega_2, &
w_3 = \omega_3.
\end{array}
\end{array}
\end{equation*}
Она приводит интегралы к виду
\begin{equation*}\label{eq1.6}
\begin{array}{l}
\displaystyle{ H = \frac{1}{2}w_3^2 + w_1 w_2 - \frac{1}{2}(y_1 +
y_2 )
},  \\
\displaystyle{
K=(w_1^2 + x_1 )(w_2^2  + x_2 )},  \\
\displaystyle{ G = \frac{1}{4}(p^2  - x_1 x_2 )w_3^2
+\frac{1}{2}(x_2 z_1 w_1 +x_1 z_2 w_2 )w_3  +}  \\
\displaystyle{\qquad{} + \frac{1}{4}(x_2 w_1  + y_1 w_2 )(y_2 w_1 +
x_1 w_2 ) - \frac{1}{4}p^2 (y_1  + y_2 ) +\frac{1} {4}r^2 (x_1 + x_2
)}.
\end{array}
\end{equation*}
Здесь введены параметры $p>r>0$
\begin{equation*}\label{ne1_25}
\begin{array}{c}
p^2  = a^2  + b^2 ,\;r^2  = a^2  - b^2,
\end{array}
\end{equation*}
с использованием которых уравнения геометрических интегралов
(\ref{eq1.4}) запишутся так:
\begin{equation*}\label{ne1_26}
\begin{array}{c}
z_1^2  + x_1 y_2  = r^2 ,\quad z_2^2  + x_2 y_1  = r^2 , \\
x_1 x_2  + y_1 y_2  + 2z_1 z_2  = 2p^2 .
\end{array}
\end{equation*}

Критическое множество отображения (\ref{eqmom}) состоит из четырех
критических подсистем. За вычетом множества особенностей нулевой
меры три из них являются критическими многообразиями ранга 2, и одно
-- критическим многообразием ранга 1. Отклоняясь от формальной
строгости, будем говорить, что первые три -- это подсистемы с двумя
степенями свободы, а последняя -- подсистема с одной степенью
свободы.

Первая критическая подсистема $\mm \subset \mP ^6$ была найдена в
работе \cite{Bogo} еще до открытия факта интегрируемости системы
(\ref{eq1.1}) в целом. Она задана соотношениями
\begin{equation}
\label{eq1.7} Z_1  = 0,\qquad Z_2  = 0,
\end{equation}
где
\begin{equation*}\label{ne1_28}
Z_1=w_1^2+x_1, \qquad Z_2=w_2^2+x_2.
\end{equation*}
Соответствующая часть $\Sigma_1$ бифуркационной диаграммы лежит в
плоскости $k = 0$, а независимыми интегралами на $\mm$ можно выбрать
$H$ и частный интеграл Богоявленского
\begin{equation*}\label{fbogo}
F = w_1 w_2 w_3+z_2 w_1+z_1 w_2.
\end{equation*}
На $\mm$ постоянные $g,k$ интегралов $G,K$ выражаются через
постоянные $h,f$ принятых за независимые интегралов $H,F$
соотношениями, найденными в \cite{Zot}.  Их можно принять за
параметрические уравнения несущей поверхности
\begin{equation}\label{gbogo}
\pov_1: \left\{ \begin{array}{l}
k=0, \\
g = \ds{\frac{1}{2}}p^2 h -\ds{\frac{1}{4}} f^2.
\end{array}\right.
\end{equation}

Условимся о некоторых упрощениях обозначений. Рассмотрим на $\mm$
"индуцированное"\ отображение момента
\begin{equation*}
\mF_1 = H{\times}F^2: \mm \to \bR ^ 2.
\end{equation*}
Квадрат интеграла $F$ взят здесь, чтобы получить однозначное (хотя
бы почти всюду) соответствие точек допустимых областей поверхности и
плоскости  в силу уравнений (\ref{gbogo}). Пользуясь этими
уравнениями не будем различать $\mF|_{\mm}$ и $\mF_1$. Множество
$\Sigma_1$ отобразится на плоскости $(h,f^2)$ как $\Im \mF_1$, а
бифуркационная диаграмма отображения $\mF_1$ отобразится на
поверхности $\pov_1$ как определенный выше остов $\Sigma_1^*$
клеточного комплекса $\Sigma_1$. В связи с этим для этой и
последующих подсистем обозначаем одинаково множества $\Sigma_i$,
$\Sigma_i^*$ и их образы на плоскостях констант интегралов,
выбранных на $\mathcal{M}_i$ в качестве независимых.

Вторая критическая подсистема $\mn  \subset \mP ^6$ найдена в
\cite{Kh32} и может быть определена как замыкание множества решений
системы уравнений
\begin{equation}\label{eq1.8}
F_1  = 0, \qquad F_2  = 0,
\end{equation}
где
\begin{equation*}\label{ne1_32}
\begin{array}{l}
F_1  = \sqrt{x_1 x_2} w_3  - \ds{\frac{(x_2 z_1 w_1  + x_1 z_2
w_2)}{\sqrt{x_1 x_2}}},\qquad \displaystyle{F_2
=\frac{x_2}{x_1}Z_1-\frac{x_1}{x_2}Z_2}.
\end{array}
\end{equation*}
Соответствующая часть $\Sigma_2$ бифуркационной диаграммы лежит на
поверхности
\begin{equation}\label{eqpov2}
(p^2h-2g)^2-r^4k=0.
\end{equation}
В качестве независимых интегралов можно взять $H$ и частный интеграл
\begin{equation*}\label{eq1.9}
\displaystyle{M =
\frac{1}{2r^2}(\frac{x_2}{x_1}Z_1+\frac{x_1}{x_2}Z_2)},
\end{equation*}
так что отображение момента для этой системы
\begin{equation*}\label{fharl1}
\mF_2 = M{\times}H: \mn \to \bR ^ 2.
\end{equation*}
При этом получаем уравнения несущей поверхности
\begin{equation*}\label{gharl}
\pov_2: \left\{ \begin{array}{l} k=r^4 m^2, \\
g = \ds{\frac{1}{2}} (p^2 h - r^4 m).
\end{array}\right.
\end{equation*}

Третья критическая подсистема $\mo  \subset \mP ^6$ найдена в
\cite{Kh34} и может быть определена как замыкание множества решений
системы уравнений
\begin{equation}\label{eq1.10}
R_1  = 0,\qquad R_2  = 0,
\end{equation}
где
\begin{equation*}\label{ne1_36}
\begin{split}
R_1 & =\displaystyle{\frac{w_2 x_1+w_1 y_2+w_3 z_1}{w_1}-\frac{w_1
x_2+w_2 y_1+w_3
z_2}{w_2},} \\
R_2 & = \displaystyle{(w_2 z_1+w_1 z_2)w_3^2+\Bigl[\frac{w_2
z_1^2}{w_1}+\frac{w_1 z_2^2}{w_2}+w_1 w_2(y_1+y_2)+}\\
& \phantom{=} + x_1 w_2^2+x_2 w_1^2\Bigr]w_3 +\frac{w_2^2 x_1
z_1}{w_1} + \frac{w_1^2 x_2
z_2}{w_2}+\\
& \phantom{=} + \displaystyle{ x_1 z_2 w_2+ x_2 z_1 w_1 +(w_1
z_2-w_2 z_1)(y_1-y_2).}
\end{split}
\end{equation*}
Соответствующая часть $\Sigma_3$ бифуркационной диаграммы лежит на
дискриминантной поверхности многочлена, обобщающего второй многочлен
Ковалевской:
\begin{equation*}\label{eq2_15}
\varphi(s) = s^4 - 2h s^3 +(h^2 + p^2 - k)s^2  - 2g s + a^2 b^2.
\end{equation*}
При этом $s$ есть константа частного интеграла
\begin{equation}\label{eq1.11}
\begin{array}{l}
\displaystyle{S=-\frac{1}{4} \big( \frac {y_2 w_1+x_1 w_2+z_1
w_3}{w_1}+\frac{x_2 w_1+y_1 w_2+z_2 w_3}{w_2} \big),}
\end{array}
\end{equation}
который в дополнение к $H$ может быть взят в качестве независимого
для определения отображения момента
\begin{equation*}\label{fharl3}
\mF_3 = S{\times}H: \mo \to \bR ^ 2.
\end{equation*}
Тогда
\begin{equation}\label{ghar2}
\pov_3: \left\{ \begin{array}{l}
\displaystyle{k = 3 s^2 - 4 h s + p^2 + h^2 - \frac{a^2 b^2} {s^2}}, \\
\displaystyle{g = -s^3 + h s^2 + \frac{a^2 b^2}{s}}.
\end{array}\right.
\end{equation}

При построении явных решений системы $\mo$ оказалось удобным
рассмотреть вместо пары интегралов $(S,H)$ пару $(S,T)$, где $T$ так
же, как и $S$, есть частный интеграл на $\mo$:
\begin{equation}\label{intT}
\begin{array}{l}
\displaystyle{T=\frac{1}{2}[w_1(x_2 w_1+y_1 w_2+z_2 w_3)+w_2(y_2
w_1+x_1 w_2+z_1 w_3)]+x_1 x_2+z_1 z_2}.
\end{array}
\end{equation}
Обозначая его постоянную через $\tau$, представим параметрические
уравнения (\ref{ghar2}) поверхности $\pov_3$ в виде
\begin{equation}\label{ne1_38}
\pov_3: \left\{ \begin{array}{l} \ds{h = \frac {p^2 -\tau}{2s}+ s}, \\
\ds{k= \frac{\tau^2 - 2p^2\tau + r^4}{4s^2}+\tau},\\
\ds{g = \frac{p^4-r^4}{4s}+ \frac{1}{2}(p^2 -\tau)s}.
\end{array}\right.
\end{equation}

Несмотря на то что уравнения (\ref{eq1.10}) имеют очевидную
особенность, замыкание множества их решений, будучи инвариантным
относительно фазового потока (\ref{eq1.1}), содержит в себе и часть
траекторий, на которых
\begin{equation*}\label{pend3}
w_1  \equiv 0,\qquad w_2  \equiv 0.
\end{equation*}
В целом семейство траекторий, удовлетворяющих этому условию,
образует четвертую критическую подсистему $\ml$, фазовое
пространство которой задано в $\mP^6$ системой уравнений
\begin{equation*}\label{crit4}
w_1 =0,\quad w_2=0, \quad z_1=0,\quad z_2=0.
\end{equation*}
Очевидно, $\ml$ -- двумерное гладкое инвариантное многообразие, $H$
выступает как единственный независимый интеграл, и, соответственно,
его константа есть единственный параметр на поверхности, несущей
соответствующую часть $\Sigma_4$ бифуркационной диаграммы. Это пара
прямых
\begin{equation*}\label{sigma4}
\pov_4: \left\{ \begin{array}{l} k=(a \mp b)^2,\\
g=\pm a b h.
\end{array}\right.
\end{equation*}
Ввиду того, что множества $\mo, \ml$ имеют непустое пересечение,
фактически $\pov_4$ добавляет в бифуркационную диаграмму $\Sigma$
лишь сегмент
\begin{equation*}\label{seg4}
g=  a b h, \quad k=(a-b)^2, \quad h^2 < 4 a b.
\end{equation*}
В остальном, как легко видеть, точки $\pov_4$ являются точками
самопересечения поверхности $\pov_3$.

В работе \cite{Kh2005} доказано, что
\begin{equation*}
\mK = \bigcup_{i=1}^4 \mathcal{M}_i,
\end{equation*}
и приведены бифуркационные диаграммы некоторых отображений, составленных из частных интегралов. Полное описание множеств $\Sigma_i$, основанное на стратификации критического множества рангом
отображения момента, дано в \cite{Kh36}. Коротко перечислим основные факты.

Обозначим через $\mbf{e}_i$ $(i=1,2,3)$ канонический базис в
$\bR^3$. Множество $\mK^0$ состоит из неподвижных точек
системы~(\ref{eq1.1})
\begin{equation}\label{immov}
\begin{array}{llll}
c_0:& {\boldsymbol\omega}={\boldsymbol 0}, &  {\bs \alpha} = \phantom{-}a \, \mbf{e}_1, & {\bs \beta}=\phantom{-}b \,\mbf{e}_2,\\
c_1:& {\boldsymbol\omega}={\boldsymbol 0}, &  {\bs \alpha} = \phantom{-}a \,\mbf{e}_1, & {\bs \beta}=-b \,\mbf{e}_2,\\
c_2:& {\boldsymbol\omega}={\boldsymbol 0}, &  {\bs \alpha} = -a \,\mbf{e}_1, & {\bs \beta}=\phantom{-}b \,\mbf{e}_2,\\
c_3:& {\boldsymbol\omega}={\boldsymbol 0}, &  {\bs \alpha} = - a\,
\mbf{e}_1, & {\bs \beta}=- b \,\mbf{e}_2.
\end{array}
\end{equation}
Соответствующие значения первых интегралов при фиксированных
постоянных $a,b$ дают четыре точки в $\bR^3$:
\begin{equation*}
\label{eq:36}
\begin{array}{lll}
P_0:\quad
\displaystyle{g=-ab(a+b),} & \displaystyle{k=(a-b)^2,} & \displaystyle{h=-(a+b);}\\[1.5mm]
P_1:\quad
\displaystyle{g=ab(a-b),} & \displaystyle{k=(a+b)^2,} & \displaystyle{h=-(a-b);}\\[1.5mm]
P_2:\quad
\displaystyle{g=-ab(a-b),} & \displaystyle{k=(a+b)^2,} & \displaystyle{h=a-b;}\\[1.5mm]
P_3:\quad \displaystyle{g=ab(a+b), } &
\displaystyle{k=(a-b)^2,} &
\displaystyle{h=a+b.}
\end{array}
\end{equation*}
В таблице~\ref{tab31} приводится соответствие между $P_k$ и точками
множеств $\Sigma_2$ и $\Sigma_3$ -- узловыми точками бифуркационных
диаграмм $\Sigma_2^*$ и $\Sigma_3^*$ в терминах выбранных там
констант интегралов. На множество $\Sigma_1$, как отмечено в
\cite{Zot}, эти точки в неприводимом случае (\ref{aneb}) не
попадают. На множестве $\Sigma_4$ их расположение очевидно.

{\renewcommand{\arraystretch}{1.5} \setlength{\extrarowheight}{-2pt}
\begin{table}[!htbp]
\centering
\begin{tabular}{|c|c|c|l|}
\multicolumn{4}{r}{\fts{Таблица
\myt\label{tab31}}}\\
 \hline $\mK^0$ &
\begin{tabular}{c}
Образ в\\
 ${\bR}^3(h,k,g)$
 \end{tabular}&
  \begin{tabular}{c}
Образ в ${\bR}^2(m,h)$
 \end{tabular}  &
\multicolumn{1}{c|}{\begin{tabular}{c} Образ в ${\bR}^2(s,h)$
 \end{tabular}} \\
\hline $c_0$  &$P_0$
&$P_{01}\left(\ds{\frac{1}{a+b}},-(a+b)\right)$&
\begin{tabular}{l}
$Q_{01} (-a,-(a+b))$\\
$Q_{02}(-b,-(a+b))$
\end{tabular}\\
\hline $c_1$  &$P_1$
&$P_{11}\left(\ds{\frac{1}{a-b}},-(a-b)\right)$&
\begin{tabular}{l}
$Q_{11}(-a,-(a-b))$\\
$Q_{12}(b,-(a-b))$
\end{tabular}\\
\hline $c_2$ &$P_2$ &$P_{21}\left(-\ds{\frac{1}{a-b}},\,
a-b\right)$&
\begin{tabular}{l}
$Q_{21}(-b,a-b)$\\
$Q_{22}(a,a-b)$
\end{tabular}\\
\hline $c_3$ &$P_3$ &$P_{31}\left(-\ds{\frac{1}{a+b}},\, a+b\right)$
&
\begin{tabular}{l}
$Q_{31}(b,a+b)$\\
$Q_{32}(a,a+b)$
\end{tabular}\\
\hline
\end{tabular}\,
\end{table}
}

В составе множества $\mK^1$ имеются следующие шесть семейств
маятниковых движений \cite{Kh34} (первый индекс соответствует
верхнему знаку):
\begin{equation*}
\begin{array}{l}
{\mL}_{1,2}=\{{\boldsymbol \alpha } \equiv \pm a{\bf e}_1, \;
{\boldsymbol \beta } = b({\bf e}_2 \cos \theta - {\bf e}_3 \sin
\theta ), \; {\boldsymbol \omega } = \theta ^ {\boldsymbol \cdot}
{\bf e}_1 , \; 2\theta ^{ {\boldsymbol \cdot}  {\boldsymbol \cdot} }
=  - b\sin \theta\},\\
{\mL}_{3,4}=\{{\boldsymbol \alpha } = a({\bf e}_1 \cos \theta + {\bf
e}_3 \sin \theta ), \; {\boldsymbol \beta } \equiv  \pm b{\bf e}_2 ,
\; {\boldsymbol \omega } = \theta ^ {\boldsymbol \cdot}  {\bf e}_2 ,
\; 2\theta ^{ {\boldsymbol \cdot}  {\boldsymbol \cdot} }  = -
a\sin \theta\}, \\
{\mL}_{5,6}=\{{\boldsymbol{\alpha }} = a({\bf{e}}_1 \cos \theta -
{\bf{e}}_2 \sin \theta ),\;{\boldsymbol{\beta }} =  \pm b({\bf{e}}_1
\sin \theta  + {\bf{e}}_2 \cos \theta ), \; {\boldsymbol{\omega }} =
\theta ^ {\boldsymbol \cdot}  {\bf{e}}_3 ,\; \\
\qquad\,\,\quad\theta ^{ {\boldsymbol
\cdot} {\boldsymbol \cdot} }  =  - (a \pm b)\sin \theta\}.
\end{array}
\end{equation*}
Семействам ${\mL}_j$ отвечают значения первых интегралов,
заполняющие кривые $\lambda_j$ в составе одномерного остова
бифуркационной диаграммы:
\begin{equation*}
\begin{array}{l}
\lambda_{1,2}=\{g = a^2 h\pm a r^2,k=(h\pm
2a)^2,h \geqslant \mp(a\pm b)\},\\
\lambda_{3,4}=\{g = b^2 h\mp b r^2, k=(h\pm
2b)^2, h \geqslant -(a\pm b)\},\\
\lambda_{5,6}=\{g=\pm abh,k=(a\mp b)^2,h\geqslant -(a\pm b)\}.
\end{array}
\end{equation*}
В таблице~\ref{tab32} приводится соответствие между кривыми
$\lambda_j$ и одномерными остовами диаграмм $\Sigma_2^*$ и
$\Sigma_3^*$, при этом сами кривые разбиты на участки $\lambda_{j
i}$ значениями интегралов на уровнях, содержащих неподвижные точки
системы. На этих же уровнях могут находиться и движения в составе
семейств ${\mL}_j$, асимптотические к неподвижным точкам.

{\renewcommand{\arraystretch}{1.5} \setlength{\extrarowheight}{-2pt}
\begin{table}[!htbp]
\centering
\begin{tabular}{|c|l|c|c|}
\multicolumn{4}{r}{\fts{Таблица
\myt\label{tab32}}}\\
 \hline
$\mK^1$& \multicolumn{1}{c|}{
\begin{tabular}{c}
Образ в\\
 ${\bR}^3(h,k,g)$
 \end{tabular}}&
 \begin{tabular}{c}
Образ в\\
 ${\bR}^2(m,h)$
\end{tabular} &
\begin{tabular}{c}
Образ в\\
 ${\bR}^2(s,h)$
\end{tabular} \\
\hline ${\mL}_1$ &\begin{tabular}{l}
$\lambda_1=\lambda_{11}\cup\lambda_{12},$\\
$\lambda_{11}:-(a+b)<h<-(a-b),$\\
$\lambda_{12}:h>-(a-b),$
\end{tabular}
&$h=r^2m-2a$ &
$s=-a$\\
\hline ${\mL}_2$ &\begin{tabular}{l}
$\lambda_2=\lambda_{21}\cup\lambda_{22}\cup\lambda_{23},$\\
$\lambda_{21}:a-b<h<a+b$,\\
$\lambda_{22}:a+b<h<2a$,\\
$\lambda_{23}: h>2a$
\end{tabular}
 &$h=r^2m+2a$&
$s=a$\\
\hline ${\mL}_3$ &
\begin{tabular}{l}
$\lambda_3=\lambda_{31}\cup\lambda_{32}\cup\lambda_{33},$\\
$\lambda_{31}:-(a+b)<h<-2b$,\\
$\lambda_{32}:-2b<h<a-b$,\\
$\lambda_{33}: h>a-b$
\end{tabular}
&$h=-r^2m-2b$& $s=-b$
\\
\hline ${\mL}_4$ &
\begin{tabular}{l}
$\lambda_4=\lambda_{41}\cup\lambda_{42}\cup\lambda_{43},$\\
$\lambda_{41}:-(a-b)<h<2b$,\\
$\lambda_{42}:2b<h<a+b$,\\
$\lambda_{43}: h>a+b$
\end{tabular}
&$h=-r^2m+2b$ &
$s=b$\\
\hline ${\mL}_5$ &
\begin{tabular}{l}
$\lambda_5=\bigcup_{j=1}^4\,\lambda_{5 j},$\\
$h > -(a+b), h^2 > 4 ab$
\end{tabular}
&--&
\begin{tabular}{l}
$\lambda_{51}: s\in(-a,-b)$,\\
$\lambda_{52}:s\in(0,b)$,\\
$\lambda_{53}:s\in(b,a)$,\\
$\lambda_{54}:s\in(a,+\infty)$,\\
$h=s+\frac{ab}{s}$
\end{tabular}
\\
\hline ${\mL}_5$
&\begin{tabular}{l}$\lambda_{5 0}, -2\sqrt{ab}<h<2\sqrt{ab}$\end{tabular}&--&--\\
\hline  ${\mL}_6$ &
\begin{tabular}{l}
$\lambda_6=\bigcup_{j=1}^4\,\lambda_{6j},$\\
$h > -(a-b)$
\end{tabular}
&-- &\begin{tabular}{l}
$\lambda_{61}:s\in(-a,-b)$,\\
$\lambda_{62}:s\in(-b,0)$,\\
$\lambda_{63}:s\in(b,a)$,\\
$\lambda_{64}:s\in(a,+\infty)$,\\
$h=s-\frac{ab}{s}$
\end{tabular}\\
\hline
\end{tabular}\,
\end{table}
}
Оставшуюся часть множества $\mK^1$ составляют критические движения
случая Богоявленского. Они аналитически организованы в три семейства
периодических решений, ниже обозначаемых через ${\mL}_7-{\mL}_9$
(топологически последнее семейство заполняет в $\mP^6$ два связных
двумерных многообразия) и принадлежат $\mm \cap \mo$. Первое
описание этих траекторий дано в \cite{Zot}. Здесь удобно
воспользоваться параметризацией этих семейств, указанной в
\cite{Kh37}. Она представляет собой алгебраические выражения
исходных фазовых переменных через одну вспомогательную переменную,
зависимость которой от времени выражается стандартными
эллиптическими функциями, при том, что и постоянные всех интегралов
также явно выражены через одну постоянную, а именно, через
постоянную $s$ соответствующего интеграла (\ref{eq1.11}). Имеем
следующие выражения для фазовых переменных:
\begin{equation}\label{equa:delone}
\begin{array}{l}
\omega_1=\ds{-\frac{1}{r} \sqrt{\frac{r_1 r_2 \psi_1}{2 s}}},\quad
\omega_2=-\ds{\frac{1}{r} \sqrt{-\frac{r_1 r_2 \psi_2}{2 s}}},\quad
\omega_3= \ds{\frac{\theta}{r_1+r_2} \sqrt{\frac{2 r_1
r_2}{s}}},\\[4mm]
\alpha_1=-s -\ds{\frac{r_1 r_2 \bigl[r^2+(r_1+r_2)^2
\bigr](\psi_1+\psi_2)}{4r^2(r_1+r_2)^2 s}},\\[4mm]
\alpha_2=-\ds{\frac{r_1 r_2 \bigl[r^2+(r_1+r_2)^2
\bigr]\sqrt{-\psi_1 \psi_2}}{2 r^2
(r_1+r_2)^2 s}},\quad
\alpha_3=\ds{\frac{r_1 \sqrt{-\psi_1}}{r}},\\[4mm]
\beta_1=\ds{\frac{r_1 r_2 \bigl[r^2-(r_1+r_2)^2\bigr]\sqrt{-\psi_1
\psi_2}}{2 r^2 (r_1+r_2)^2 s}},\\[4mm]
\beta_2=-s-\ds{\frac{r_1 r_2
\bigl[r^2-(r_1+r_2)^2\bigr](\psi_1+\psi_2)}{4 r^2
(r_1+r_2)^2 s}},\quad
\beta_3=\ds{\frac{r_2\sqrt{-\psi_2}}{r}},
\end{array}
\end{equation}
где
\begin{equation*}\label{psi:delone}
\begin{array}{ll}
\psi _1  = (\theta  - e'_ +  )(\theta  - e'_ -  ), & e'_ \pm   =
\displaystyle{\frac{r_1  + r_2 } {r_1 }(s \pm a)},
\\[3mm]
\psi _2  = (\theta  - e''_ +  )(\theta  - e''_ -  ), & e''_ \pm   =
\displaystyle{\frac{r_1  + r_2 } {r_2 }(s \pm b),}
\end{array}
\end{equation*}
и для различных семейств следует положить
\begin{equation*}\label{ri:delone}
\begin{array}{ll}
{\mL}_7: s \in [ - b,0), & r_1  = \sqrt {\mstrut a^2 - s^2 }
> 0, \quad r_2 = \sqrt {\mstrut b^2  - s^2 } \geqslant 0, \\[3mm]
{\mL}_8: s \in (0,b], & r_1  = \sqrt {\mstrut a^2  - s^2 }  > 0,
\quad r_2  =  - \sqrt {\mstrut b^2  - s^2 }  \leqslant 0, \\[3mm]
{\mL}_9: s \in [a, + \infty ), & \left\{\begin{array}{l} r_1  = \ri
\, r_1^*, \quad r_2  = \ri \,r_2^* \\ 0 \leqslant r_1^*  = \sqrt
{\mstrut s^2  - a^2 } < r_2^*  = \sqrt {\mstrut s^2 - b^2 }
\end{array}.\right.
\end{array}
\end{equation*}
Значения первых интегралов на этих движениях определяют еще одну
часть одномерного остова бифуркационной диаграммы $\Sigma$ в виде
трех кривых $\delta_i \subset \bR^3$ $(i=1,2,3)$, параметрические
уравнения которых получены в \cite{Kh37}
\begin{equation}\label{deltas}
\begin{array}{l}
\delta_1: \left\{
\begin{array}{l}
\displaystyle{h=2s-\frac{1}{s}\sqrt{(a^2-s^2)(b^2-s^2)}   }\\[2mm]
\displaystyle{f=\pm\sqrt{-\frac{2}{s}\sqrt{(a^2-s^2)(b^2-s^2)}}(\sqrt{a^2-s^2}+
\sqrt{b^2-s^2})}\\[3mm]
\displaystyle{\tau=(\sqrt{a^2-s^2}+\sqrt{b^2-s^2})^2}, \quad s \in
[-b,0)
\end{array} \right. ; \\
\delta_2: \left\{
\begin{array}{l}
\displaystyle{h=2s+\frac{1}{s}\sqrt{(a^2-s^2)(b^2-s^2)}   }\\[2mm]
\displaystyle{f=\pm\sqrt{\frac{2}{s}\sqrt{(a^2-s^2)(b^2-s^2)}}(\sqrt{a^2-s^2}-
\sqrt{b^2-s^2})}\\[3mm]
\displaystyle{\tau=(\sqrt{a^2-s^2}-\sqrt{b^2-s^2})^2}, \quad s \in
(0,b]
\end{array} \right. ; \\
\delta_3: \left\{
\begin{array}{l}
\displaystyle{h=2s-\frac{1}{s}\sqrt{(s^2-a^2)(s^2-b^2)}   }\\[2mm]
\displaystyle{f=\pm\sqrt{\frac{2}{s}\sqrt{(s^2-a^2)(s^2-b^2)}}(\sqrt{s^2-b^2}-
\sqrt{s^2-a^2})}\\[3mm]
\displaystyle{\tau=-(\sqrt{s^2-b^2}-\sqrt{s^2-a^2})^2}, \quad s \in
[a,+\infty)
\end{array} \right. .
\end{array}
\end{equation}
Здесь $s,\tau$ -- значения функций (\ref{eq1.11}), (\ref{intT}) на
соответствующих множествах. Как отмечено выше, $s$ выбирается за
независимый параметр. Значение $k \equiv 0$, зависимость $g(s)$
находится из (\ref{gbogo}).

Семейства ${\mL}_1 - {\mL}_4$ составляют
трансверсальное пересечение критических
многообразий $\mn \cap \mo$, а ${\mL}_7 -
{\mL}_9$ -- трансверсальное пересечение $\mm
\cap \mo$. По отношению к критическим
подсистемам с двумя степенями свободы они могут
иметь лишь эллиптический или гиперболический
тип. Из топологических результатов работ
\cite{Zot}, \cite{KhSav}, \cite{Kh2009},
\cite{Kh37}, \cite{Zot1}
 можно получить утверждения о типе этих
траекторий \textit{внутри} $\mm, \mn, \mo$. Отсюда тип этих
критических точек в трехмерной классификации можно было бы
определить по предложению~\ref{propm}. Однако эти утверждения не
подтверждены вычислениями. Поэтому мы явно выпишем необходимые
характеристические многочлены, получим тип точек этих семейств в
$\mP^6$, откуда будет следовать и доказательство утверждений по
отношению к критическим подсистемам с двумя степенями свободы.

Одномерный остов $\Sigma^1$ бифуркационной диаграммы, кроме значений
$\mF$ в точках множества $\mK^1$, может порождаться касанием листов
$\Sigma_i$. Соответствующие критические точки, хоть и имеют ранг 2,
но являются вырожденными. Доказательство этого факта требует
обоснования {\it несуществования} регулярного элемента в
соответствующей подалгебре, что всегда сложнее, чем его явное
нахождение. Здесь мы не будем приводить строгих обоснований и
вычислять тип этих точек. Ограничимся указанием уравнений
соответствующих участков на $\Sigma$ и неравенств, определяющих
область существования движений на этих участках. Имеют место
следующие случаи. Касание листов $\Sigma_1, \Sigma_2$ происходит в
точках множества
\begin{equation*}\label{d1}
\Delta_1: \left\{\begin{array}{l} k=0, \\
2g=p^2h
\end{array}\right., \quad h \geqslant -2b.
\end{equation*}
Касание листов $\Sigma_2, \Sigma_3$ происходит в точках множества
\begin{equation*}\label{d2}
\Delta_2: \left\{\begin{array}{l}
k=\ds{\frac{1}{r^4}(2g-p^2h)^2} \\
g=g_{\pm}(h)=\ds{\frac{1}{4p^2}} \left[(2p^4-r^4)h \pm r^4 \sqrt{\mstrut
h^2-2p^2} \right]
\end{array}\right., \quad h \geqslant \sqrt{\mstrut 2p}.
\end{equation*}
Кроме того, на листе $\Sigma_3$ имеется ребро возврата
("самокасание"), которое описывается параметрическими уравнениями
\begin{equation}\label{d3}
\Delta_3: \left\{\begin{array}{l}
h=\ds{\frac{p^4-r^4+12s^4}{8s^3}}\\
k=\ds{-\frac{3s^2}{4}+ p^2 -\frac{3(p^4-r^4)}{8s^2}+\frac{(p^4-r^4)^2}{64s^6}} \\
g=\ds{\frac{3(p^4-r^4)+4s^4}{8s}}
\end{array}\right., \quad 0 < s \leqslant s_0.
\end{equation}
Здесь $s_0$ -- единственный корень уравнения
\begin{equation}\label{s0}
3s^8 -
4p^2s^6+\ds{\frac{3}{2}}(p^4-r^4)s^4-\ds{\frac{(p^4-r^4)^2}{16}}=0
\end{equation}
на полупрямой $s>a$. Ниже мы явно увидим, что в характеристических
многочленах операторов, выбранных в качестве регулярных, при этих
условиях возникает кратный корень. Эти множества обладают еще одним
важным свойством, которое оказывается связанным с вырожденностью
критических точек в $\mP^6$, а именно, как показано в работах
\cite{Zot}, \cite{KhSav}, \cite{Kh2007}, на множествах
$\mathcal{M}_i \cap \mF^{-1}(\Delta_i)$ $(i=1,2,3)$ вырождается
2-форма, индуцированная на $\mathcal{M}_i$ симплектической
структурой многообразия $\mP^6$. Заметим, что внутри $\mm$
соответствующие точки регулярны, внутри $\mn, \mo$ регулярны почти
все из них (за исключением точек пересечения с $\mK^1$). Образы
множеств $\Delta_i$ на поверхностях $\pov_i$ в выбранных там
параметрах представлены в табл.~\ref{tabdel}.

{
\renewcommand{\arraystretch}{1.5}
\setlength{\extrarowheight}{-2pt}
\begin{table}[!htbp]
\centering
\begin{tabular}{|c|c|c|c|}
\multicolumn{4}{r}{\fts{Таблица
\myt\label{tabdel}}}\\
\hline & \begin{tabular}{c}Образ в\\
 ${\bR}^2(h,f^2)$
 \end{tabular}&
\begin{tabular}{c}Образ в\\
 ${\bR}^2(m,h)$
 \end{tabular}&
\begin{tabular}{c}Образ в\\
 ${\bR}^2(s,h)$
 \end{tabular}\\
\hline $\Delta_1$ &
\begin{tabular}{l}
$f=0,$\\
$h\geqslant -2b$
\end{tabular} &
\begin{tabular}{l}
$m=0,$ \\
$h\geqslant -2b$
\end{tabular} &
--\\
\hline

$\Delta_2$ & -- &
\begin{tabular}{l}
$h=-(a^2+b^2)m -\ds{\frac{1}{2m}},$\\
$m<0$\\
\end{tabular}&
\begin{tabular}{l}
$h=\ds{\frac{a^2+b^2}{2s}}+s,$\\
$s>0$
\end{tabular}\\
\hline $\Delta_3$ & -- & --
 &
\begin{tabular}{l}
$h=\ds{\frac{3s^4+a^2b^2}{2s^3}},$\\
$0<s\leqslant s_0$
\end{tabular}\\
\hline
\end{tabular}
\end{table}
}

Собирая образы точек $P_i$, одномерных множеств
$\lambda_i,\delta_i,\Delta_i$ на плоскостях значений индуцированных
отображений момента $\mF_i$, получим бифуркационные диаграммы
$\Sigma_i^*$ этих отображений (рис.~\ref{fig_sig1}--\ref{fig_sig3}).
Части $\Sigma_i$ бифуркационной диаграммы полного отображения
момента отобразятся на выбранных плоскостях как оболочка каркасов
$\Sigma_i^*$. При этом не включаются те связные компоненты
$\bR^2\backslash \Sigma_i^*$, которые помечены на рисунках символом
$\emptyset$.

\begin{figure}[!htbp]
 \centering
\includegraphics[width=0.5\textwidth,keepaspectratio]{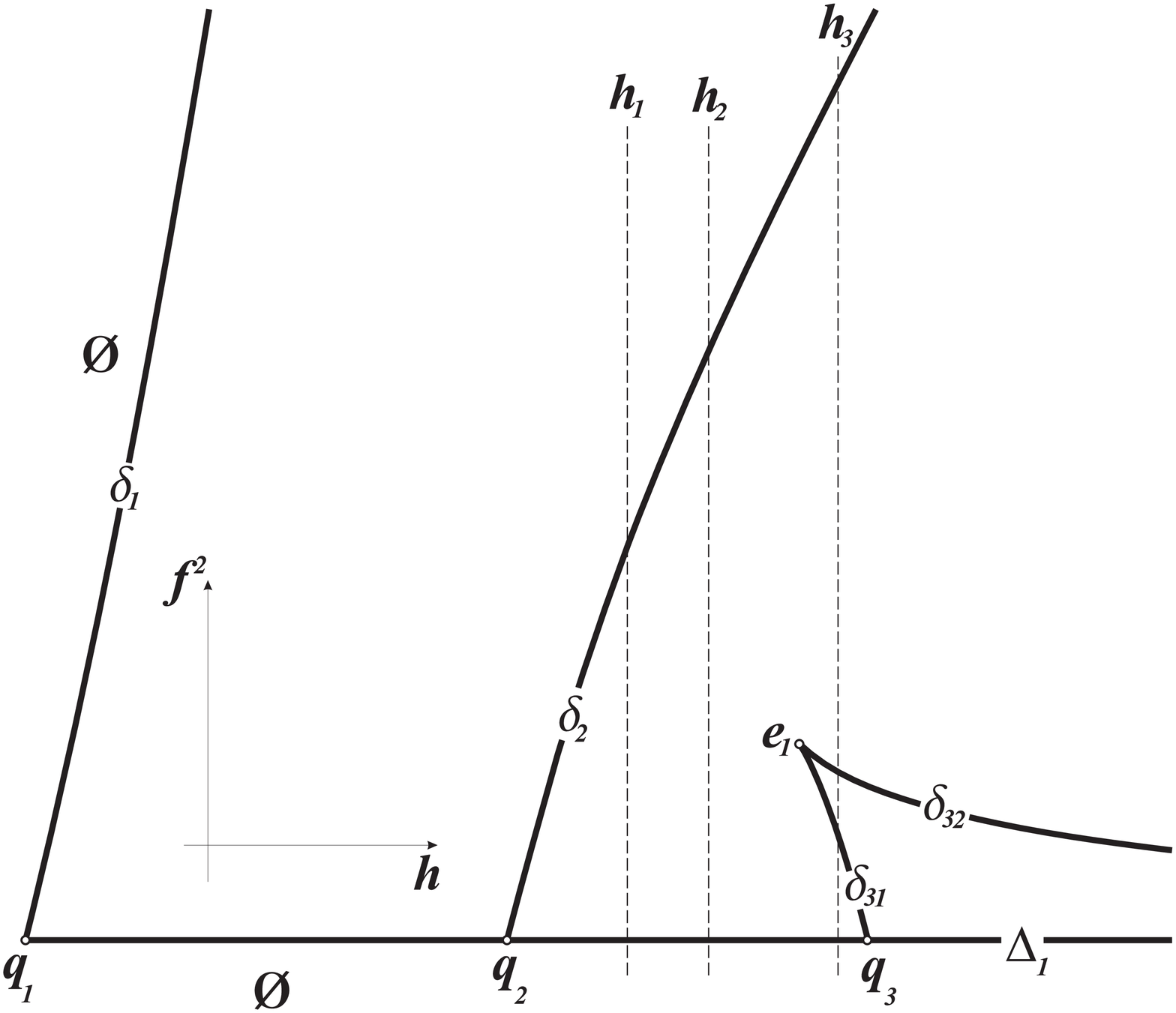}
\parbox[t]{0.9\textwidth}{\caption{Бифуркационная диаграмма
$\Sigma_1^*$  на плоскости $(h,f^2)$.}\label{fig_sig1}}
\end{figure}

\begin{figure}[!htbp]
\centering
\includegraphics[height=11cm,width=10cm,keepaspectratio]{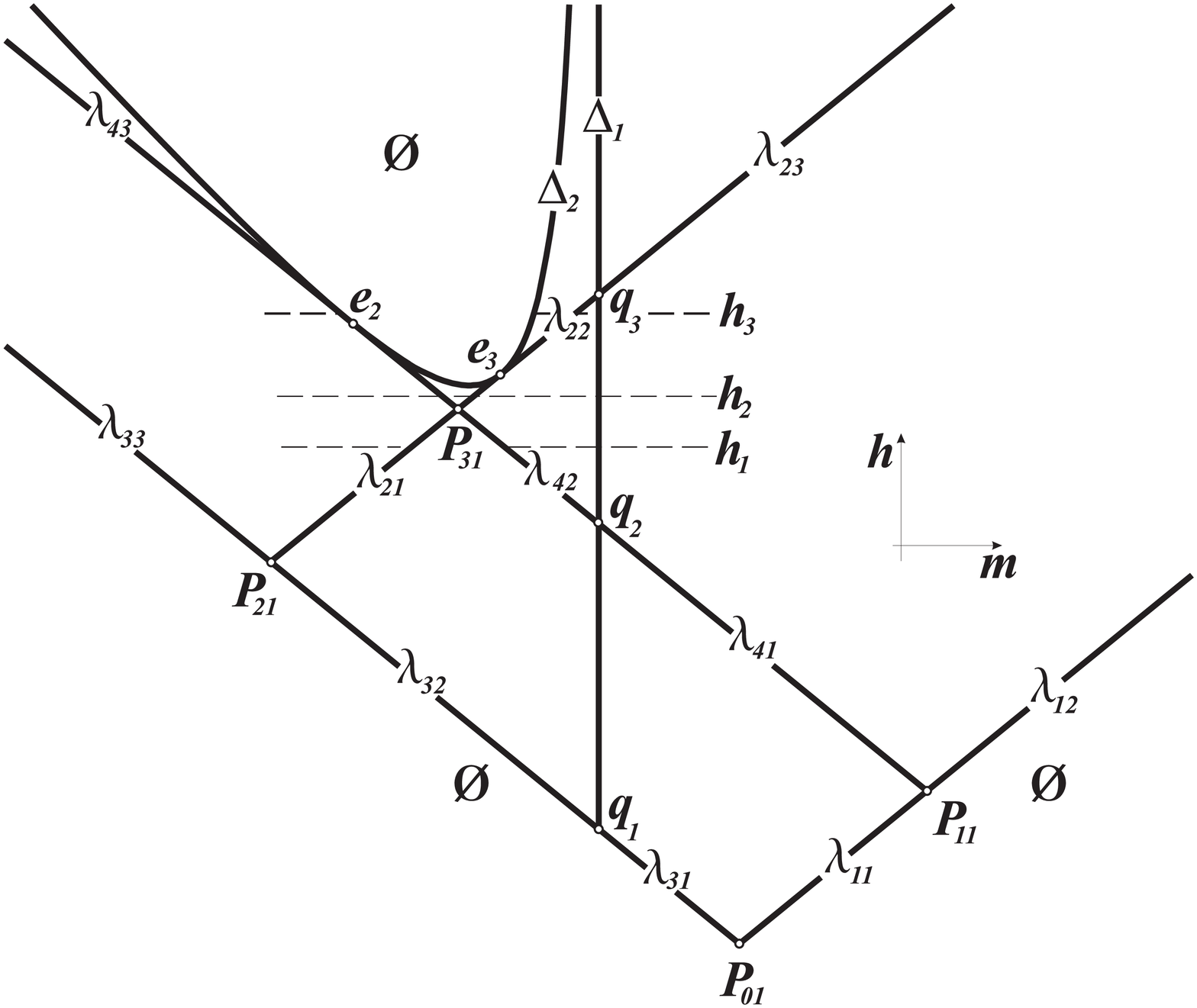}
\parbox[t]{0.9\textwidth}{\caption{Бифуркационная диаграмма
$\Sigma_2^*$ на плоскости $(m,h)$.}\label{fig_sig2}}
\end{figure}

\begin{figure}[!htbp]
 \centering
\includegraphics[height=10cm,width=10cm,keepaspectratio]{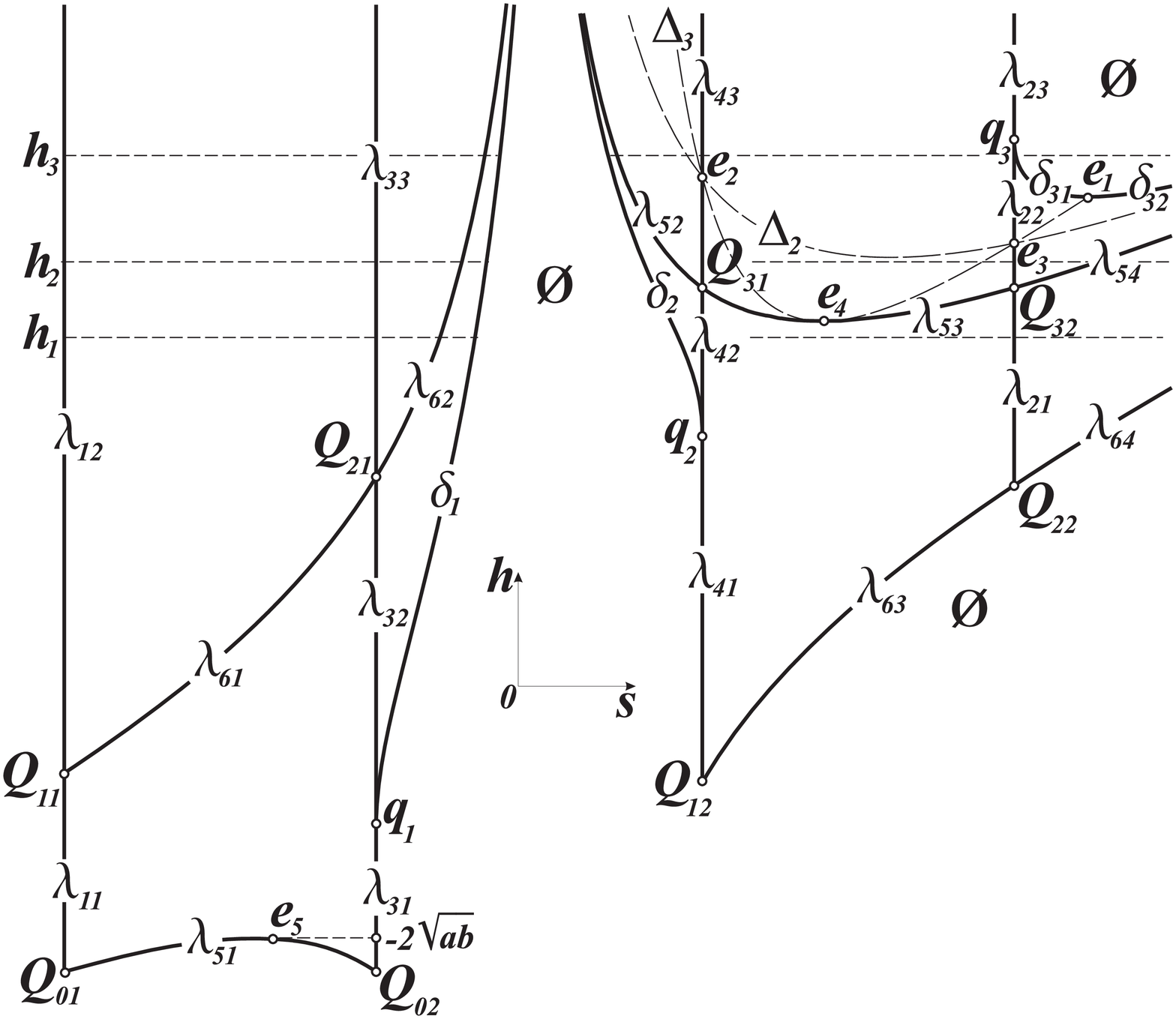}
\parbox[t]{0.9\textwidth}{\caption{Бифуркационная диаграмма
$\Sigma_3^*$  на плоскости $(s,h)$.}\label{fig_sig3}}
\end{figure}

Таким образом, $\Sigma_1^*$ состоит из кривых $\delta_1$ --
$\delta_3$ и луча $\Delta_1$ (рис.~\ref{fig_sig1}). На рисунке также
введены обозначения некоторых особых точек $q_k$, $e_k$, которые
фигурируют и на последующих рисунках. Отмечены три значения энергии
$h_1 - h_3$, для которых ниже показаны диаграммы $\Sigma(h)$.

Диаграмма $\Sigma_2^*$ (рис.~\ref{fig_sig2}) состоит из кривых
$\lambda_1 - \lambda_4$ (напомним, что мы не различаем объекты в
$\bR^3(h,k,g)$ и их образы на плоскостях значений индуцированных
отображений момента), множеств $\Delta_1$ и $\Delta_2$. Ниже нам
будет удобно ввести на $\mn $ вместо $H$ другой почти всюду
независимый с $M$ интеграл \cite{KhSav}
\begin{equation*}\label{lkhsav}
\displaystyle{L = \frac{1} {{\sqrt {x_1 x_2 } }}[w_1 w_2  + {{x_1
x_2  + z_1 z_2 }} M]}.
\end{equation*}
Его постоянная связана с $h,m$ соотношением
\begin{equation}\label{lhm}
\ell^2 = 2p^2 m^2+2 h m + 1,
\end{equation}
а множество $\Delta_2$ задано уравнением $\ell=0$.

Наконец, диаграмма $\Sigma_3^*$ (рис.~\ref{fig_sig3}) порождает
наиболее сложную часть бифуркационной диаграммы $\Sigma$.
Необходимые комментарии будут даны ниже. Пока отметим лишь, что на
ней нашли отражение практически все уже упоминавшиеся обозначения.
Например, точки $q_1 - q_3$ являются граничными точками кривых
$\delta_1 - \delta_3$ как на этом рисунке, так и на
рис.~\ref{fig_sig1}. Точка $e_1$, которая на плоскости $(h,f^2)$
является точкой возврата кривой $\delta_3$, на последнем рисунке
соответствует значению $s_0$ интеграла (\ref{eq1.11}), при котором
достигается минимальное значение $h(s)$ в соответствующем выражении
(\ref{deltas}). Оно же фигурирует в (\ref{d3}) как корень уравнения
(\ref{s0}).

Для полного описания $\Sigma \subset \mP^6$ введем некоторые
обозначения.

Обозначим обращения зависимостей $h(s)$ на кривых (\ref{deltas}):
\begin{equation*}\label{obrdel}
\begin{array}{llll}
\delta_1:& s=s_1(h), & h\in[-2b,+\infty), & s_1(h)\in
[-b,0),\\
\delta_2: & s=s_2(h), & h \geqslant 2b, & s_2(h)\in
(0,b],\\
\delta_{31}: & s=s_{31}(h), & h \in [h_0,2a], & s \in [a, s_0],\\
\delta_{32}: & s=s_{32}(h), & h \in [h_0,+\infty), & s \in
[s_0+\infty).
\end{array}
\end{equation*}
Здесь $h_0$ -- значение $h(s_0)$ на кривой $\delta_3$. Уравнение
(\ref{s0}) для вычисления значения $s_0$ теперь получим, записывая
условие минимума $h$ на кривой $\delta_3$ в виде
\begin{equation}
\label{equa:s0} 2s^2\sqrt{(s^2-a^2)(s^2-b^2)}=s^4-a^2b^2.
\end{equation}
Единственность его решения при $s>a$ очевидна.

Из соотношений (\ref{gbogo}) найдем зависимость на~$\delta_1$:
$$
g = g_1(h) = s^3+\frac{ab}{s}-s^2 \phi(s)|_{s=s_1(h)}, \quad h
\geqslant -2b.
$$
Рассматривая интервалы монотонности $h(s)$ на кривых $\lambda_5 -
\lambda_6$ обозначим:
$$
\displaystyle{s_{5-}(h)= \frac{h - \sqrt{h^2-4ab}}{2},}\quad
\displaystyle{s_{5+}(h)= \frac{h + \sqrt{h^2-4ab}}{2},}\quad
\displaystyle{s_6(h)= \frac{h + \sqrt{h^2+4ab}}{2}}.
$$
Теперь бифуркационная диаграмма полностью описывается следующей
теоремой \cite{Kh36}, которая сформулирована так, чтобы все условия
давали явные неравенства при фиксированном значении энергии $h$.

\begin{theorem}
1. Множество $\Sigma_1=\pov_1 \cap \Sigma$ имеет вид
$$
k = 0,\quad  g_1(h) \leqslant g \leqslant \frac{1}{2}p^2h, \quad h
\geqslant -2b.
$$
2. Множество $\Sigma_2=\pov_2 \cap \Sigma$ лежит в полупространстве
$h \geqslant -(a+b)$ и описывается следующей совокупностью систем
неравенств
\begin{eqnarray*}
&&\left\{ {\begin{array}{l} b^2 h -b r^2\leqslant g \leqslant a^2 h +a
r^2\\
-(a+b)\leqslant h \leqslant \sqrt{2}p
\end{array}} \right.; \\[4mm]
&&\left\{ {\begin{array}{l} b^2 h -b r^2 \leqslant g \leqslant g_{-}(h)\\
h \geqslant \sqrt{2}p
\end{array}} \right.; \\[4mm]
&&\left\{ {\begin{array}{l} g_{+}(h) \leqslant g \leqslant a^2 h +a
r^2 \\
h \geqslant \sqrt{2}p
\end{array}} \right..
\end{eqnarray*}
3. Множество $\Sigma_3=\pov_3 \cap \Sigma$ полностью описывается
следующей совокупностью условий на плоскости $(s,h)$. Для
отрицательных значений $s$:
$$
\left\{ {\begin{array}{l} -(a+b)\leqslant h \leqslant -2 \sqrt{ab}\\
s \in [-a,s_{5-}(h)] \cup [s_{5+}(h),-b]
\end{array}} \right.; \quad
\left\{ {\begin{array}{l} -2 \sqrt{ab}\leqslant h \leqslant -2b\\
s \in [-a,-b]
\end{array}} \right.; \quad
\left\{ {\begin{array}{l} h >  -2b\\
s \in [-a,s_1(h)]
\end{array}} \right..
$$
Для положительных значений $s$:
$$
\begin{array}{ll}
\left\{ {\begin{array}{l} -a+b\leqslant h \leqslant 2b\\
s \in [b,s_6(h)]\end{array}} \right.; &
\left\{ {\begin{array}{l} 2b \leqslant h \leqslant h_0\\
s \in [s_2(h),s_6(h)]
\end{array}} \right.; \\[4mm]
\left\{ {\begin{array}{l} h_0 \leqslant h \leqslant 2a\\
s \in [s_2(h),s_{31}(h)]\cup[s_{32}(h),s_6(h)]
\end{array}} \right.; &
\left\{ {\begin{array}{l} h > 2a\\
s \in [s_2(h),a]\cup[s_{32}(h),s_6(h)]
\end{array}} \right..
\end{array}
$$
4. Множество $\Sigma_4=\pov_4 \cap \Sigma$ состоит из двух лучей
\begin{equation}\notag
\begin{array}{lll}
g = abh,& k = (a - b)^2 ,& h \geqslant  - (a +
b),\\
g =  - abh,& k = (a + b)^2 ,& h \geqslant  - (a - b).
\end{array}
\end{equation}
\end{theorem}

Эта теорема дает возможность изобразить во всех деталях любую
диаграмму $\Sigma(h)$ на изоэнергетическом уровне и отследить с
помощью компьютерной графики эволюцию этих диаграмм с изменением
энергии. Разделяющими служат значения $h$, фигурирующие как границы
для неравенств в утверждении теоремы. Диаграммы для отмеченных ранее
трех значений $h_1, h_2, h_3$ с увеличенными для наглядности
участками приведены на рис.~\ref{fig1}--\ref{fig3}.
\begin{figure}[!htbp]
\centering
\includegraphics[width=6cm,keepaspectratio]{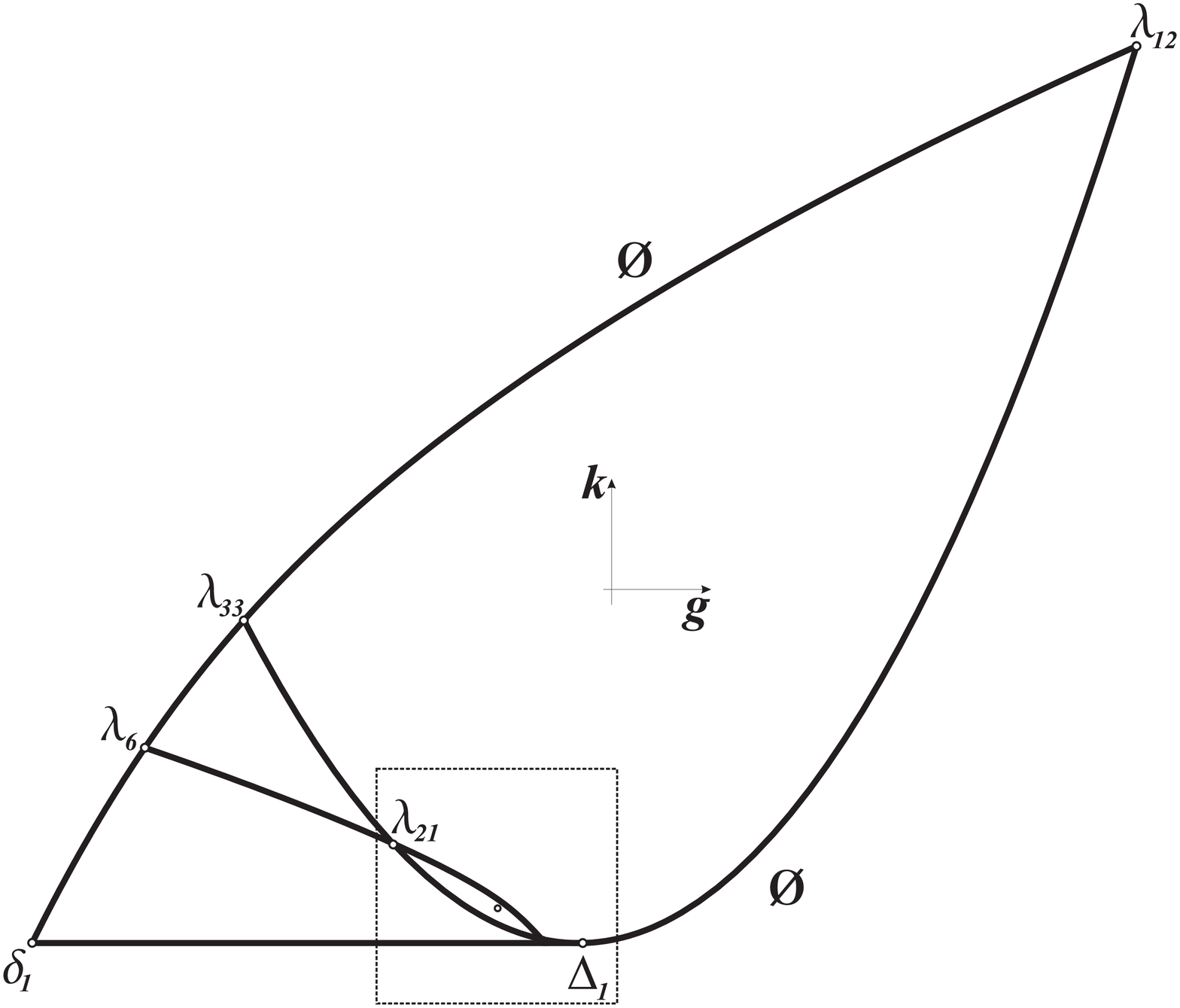}
\includegraphics[width=6cm,keepaspectratio]{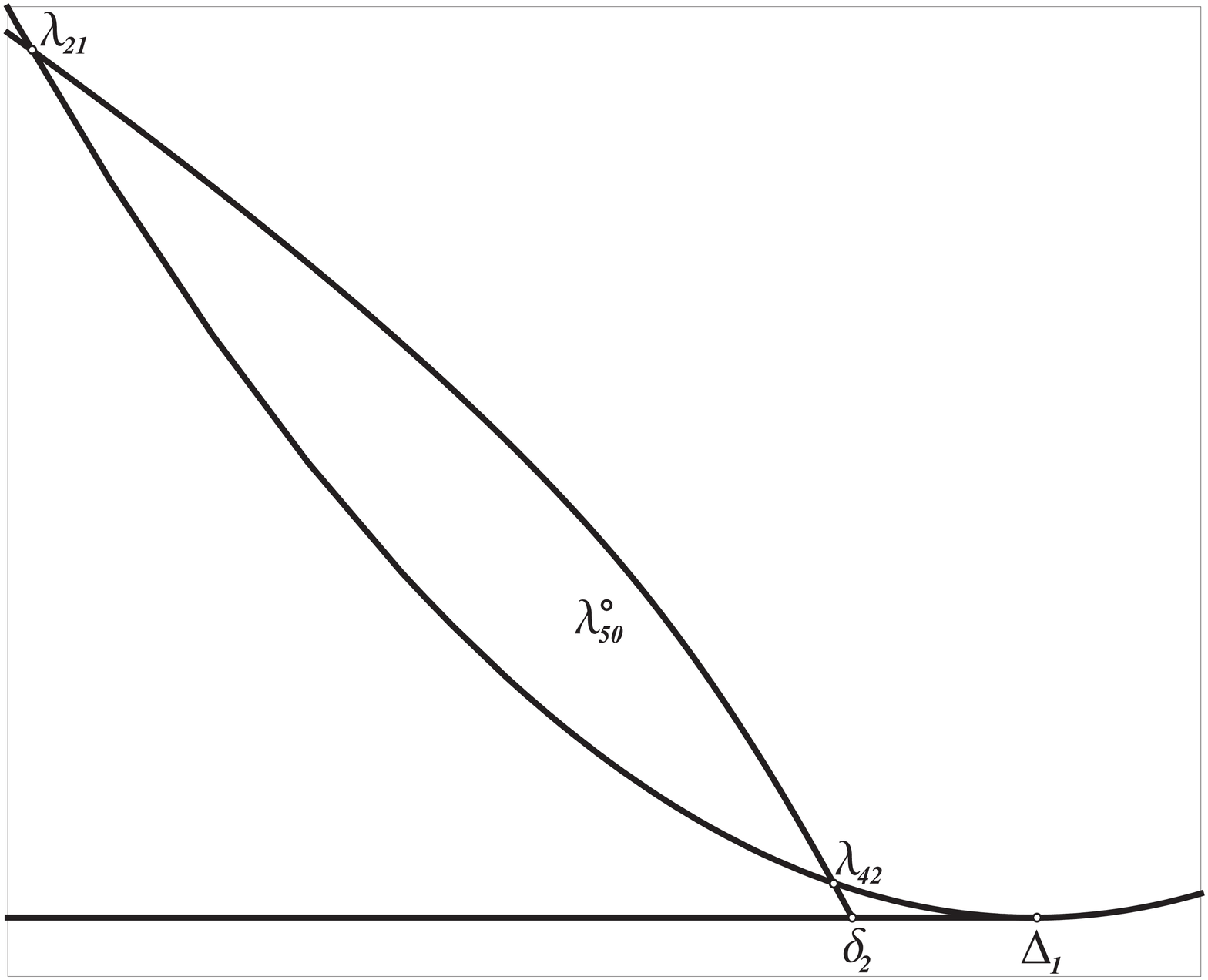}
\parbox[t]{0.9\textwidth}{\caption{Бифуркационная диаграмма
$\Sigma(h_1)$ и ее фрагмент.}\label{fig1}}
\end{figure}

\begin{figure}[!htbp]
\centering
\includegraphics[width=6cm,keepaspectratio]{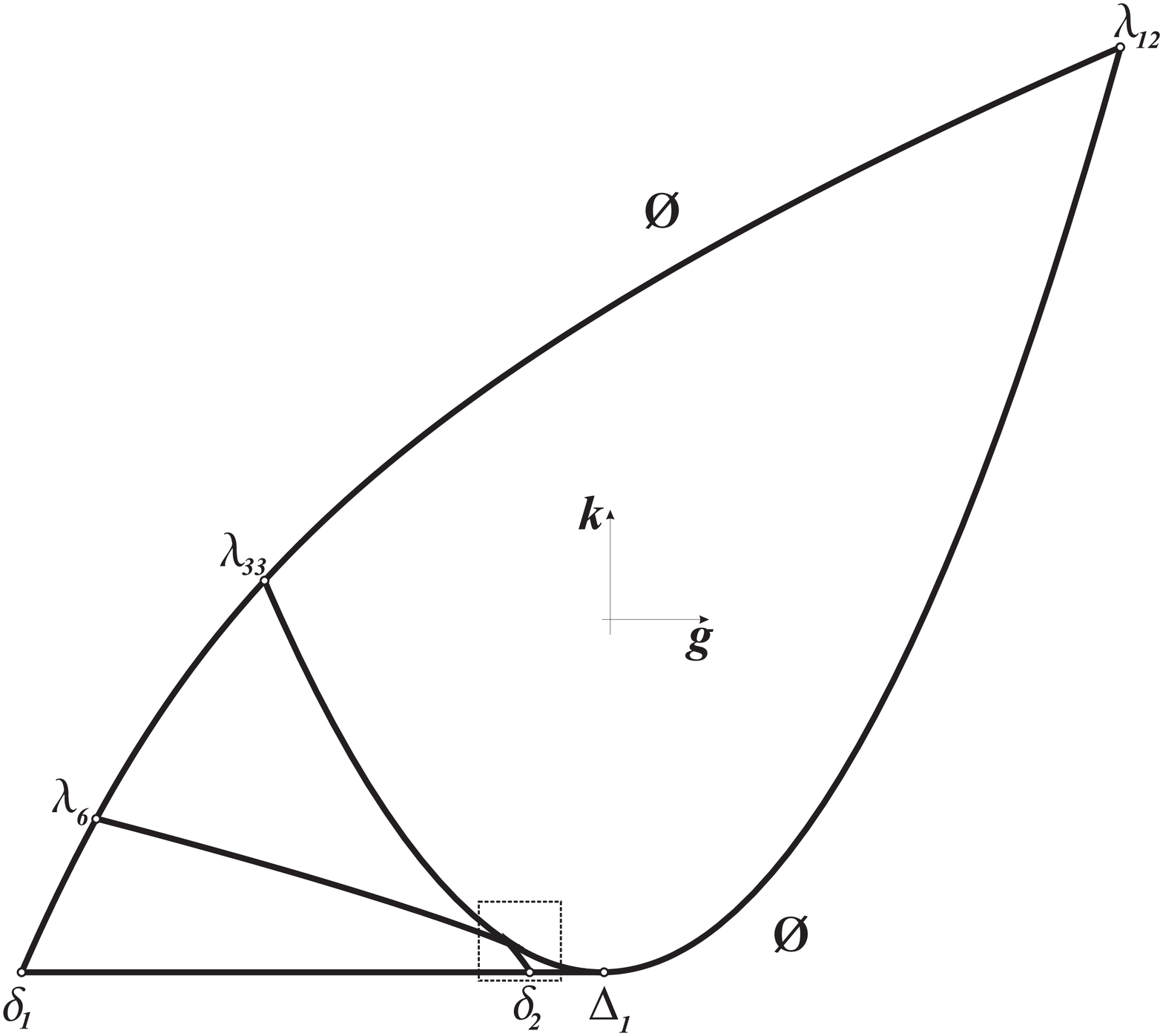}
\includegraphics[width=6cm,keepaspectratio]{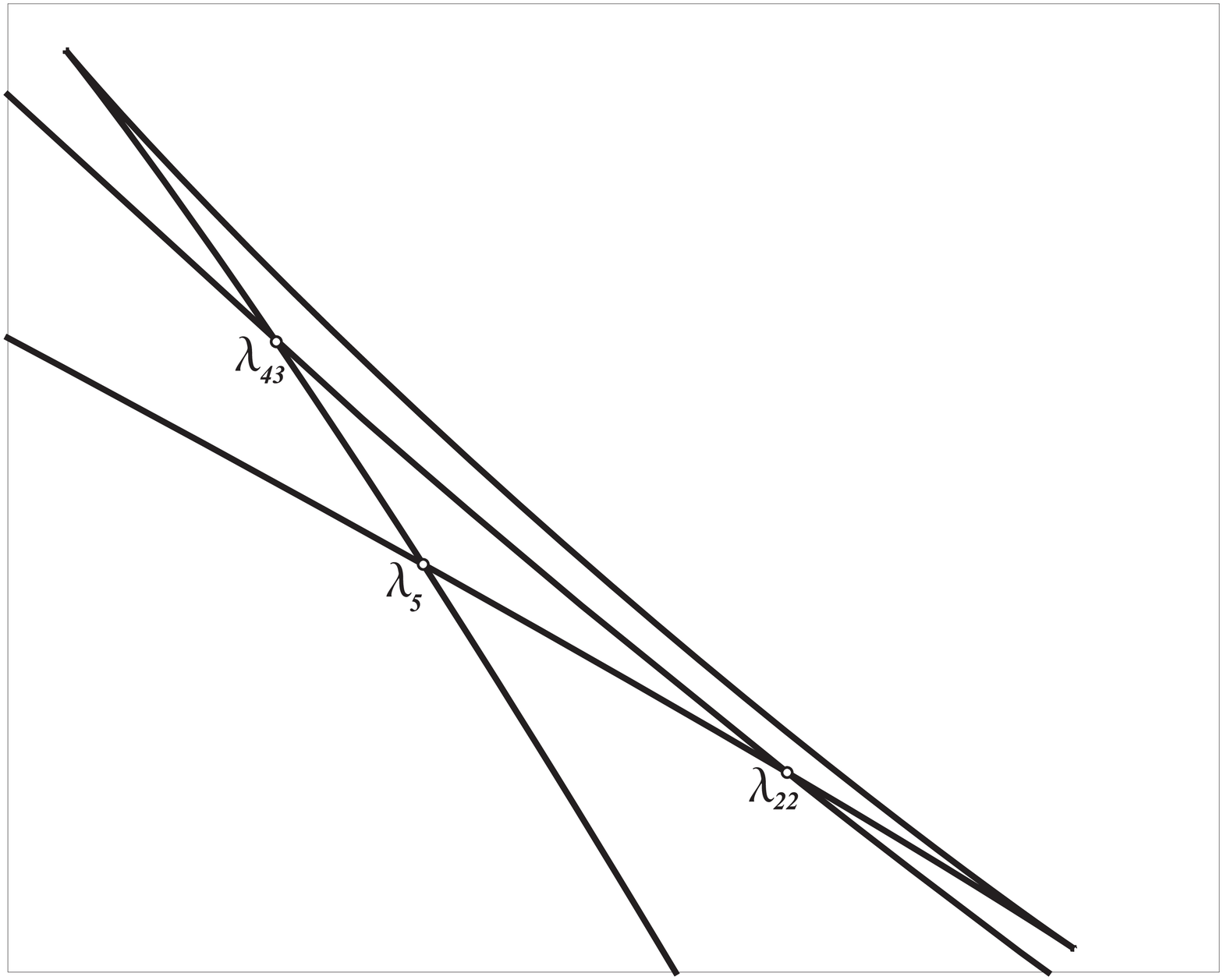}
\parbox[t]{0.9\textwidth}{\caption{Бифуркационная диаграмма
$\Sigma(h_2)$ и ее фрагмент.}\label{fig2}}
\end{figure}

\begin{figure}[!htbp]
\centering
\includegraphics[width=6cm,keepaspectratio]{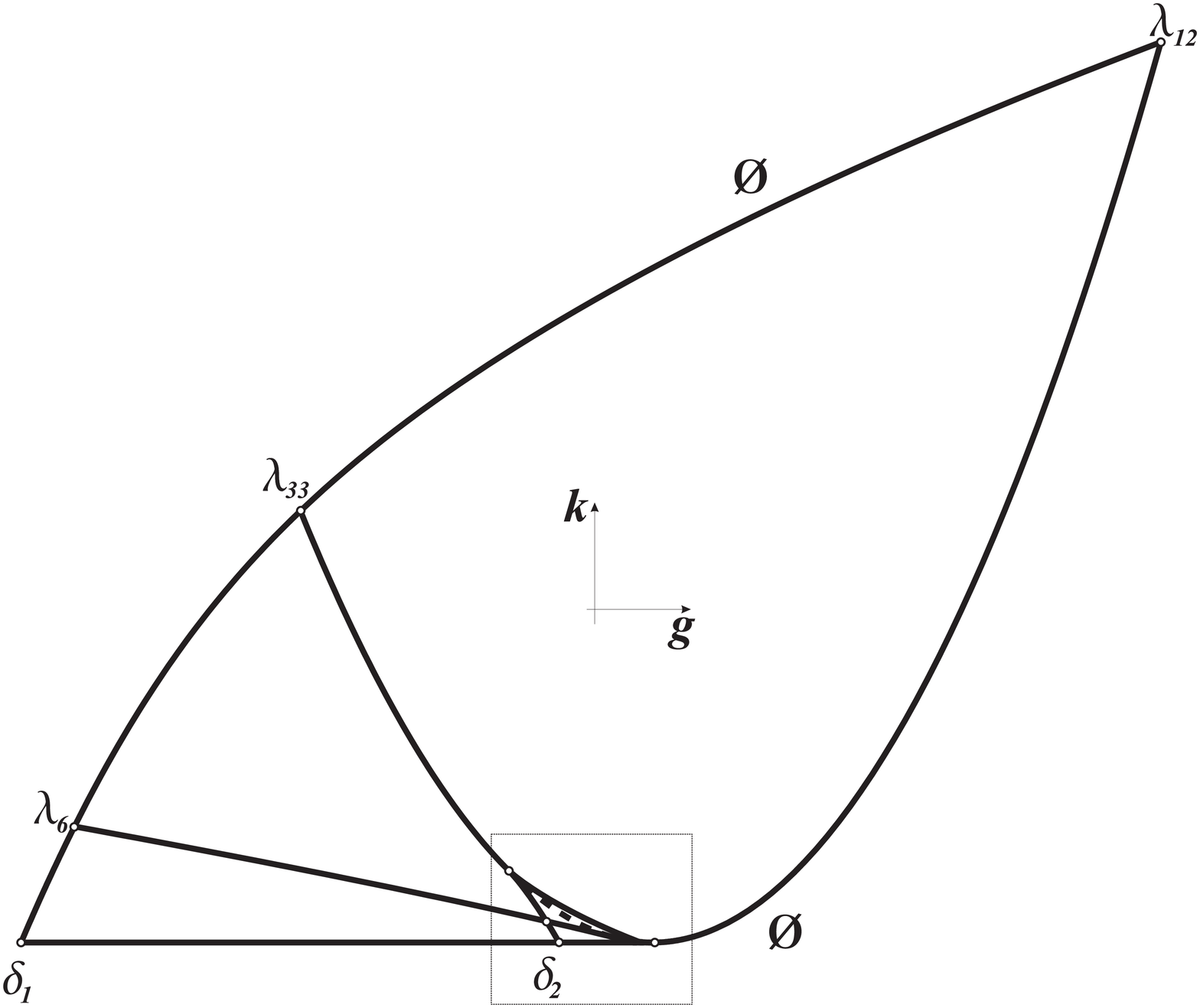}
\includegraphics[width=6cm,keepaspectratio]{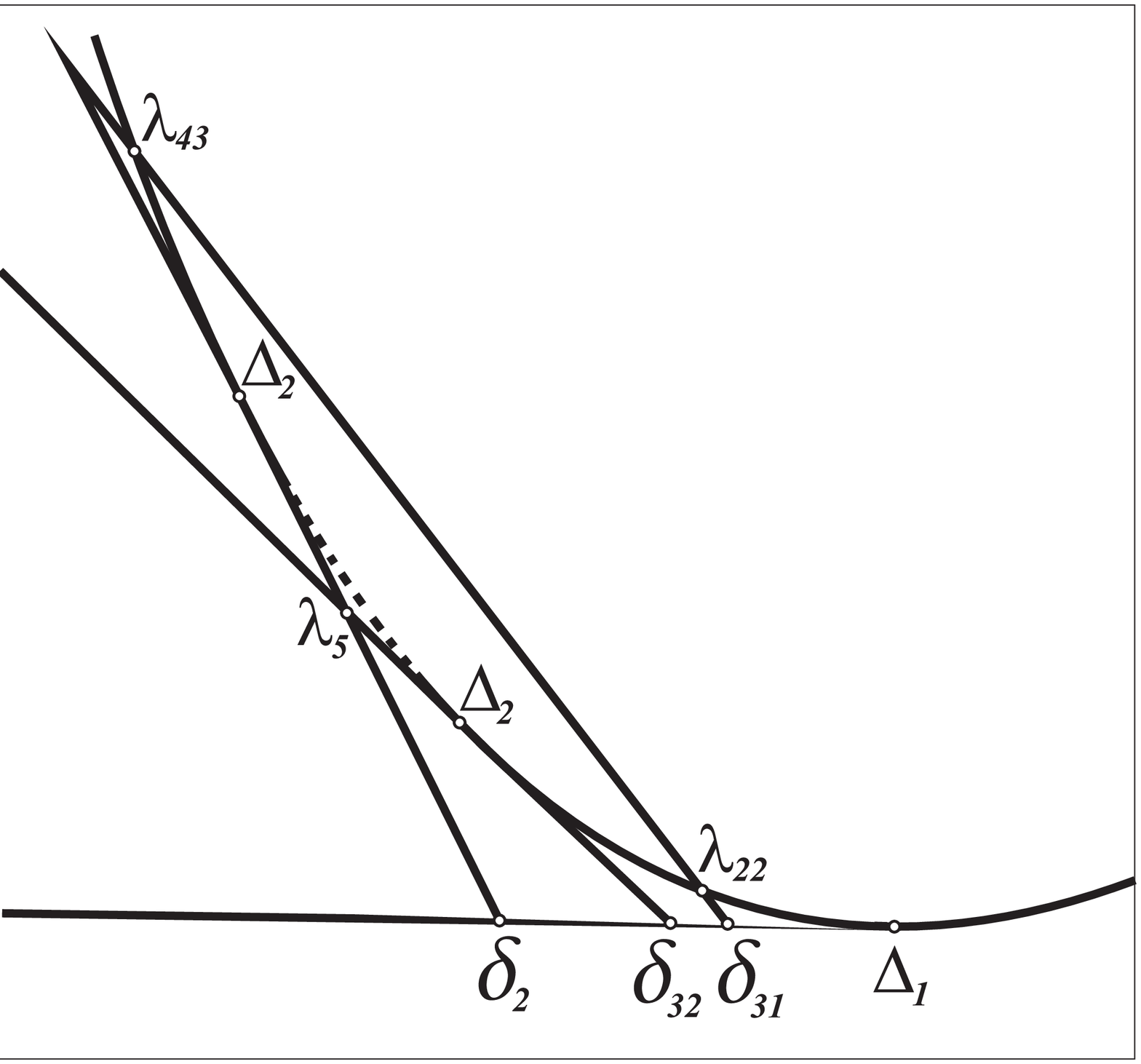}
\parbox[t]{0.9\textwidth}{\caption{Бифуркационная диаграмма
$\Sigma(h_3)$ и ее фрагмент.}\label{fig3}}
\end{figure}

\section{Классификация неподвижных точек}
В работе \cite{KhZot} найден индекс Морса гамильтониана $H$ в
неподвижных точках (\ref{immov}), что в значительной мере определяет
характер поведения системы в их окрестности. Однако строгая
классификация требует указания типа этих точек как критических точек
отображения момента.
\begin{theorem}
Особым точкам $P_k$ $(k=0,\ldots,3)$ бифуркационной диаграммы
$\Sigma$ соответствуют \textit{невырожденные} особенности $c_k$
ранга $0$ отображения момента ${\mF}$. Более точно, точке $P_0$
соответствует особенность $c_0$ типа <<центр-центр-центр>>, $P_1$ --
особенность $c_1$ типа <<центр-центр-седло>>, $P_2$ -- особенность
$c_2$ типа <<центр-седло-седло>>, $P_3$ -- особенность $c_3$ типа
<<седло-седло-седло>>.
\end{theorem}

\begin{proof}
Точки $P_k$ порождены неподвижными точками $c_k$. Касательные
пространства в точках $c_k$ можно описать в виде уравнений
\begin{eqnarray*}
&T_{c_0}\mP^6=T_{c_3}\mP^6=\{{\boldsymbol x}\in {\bR}^9: x_4=x_8=0,b x_5+ax_7=0\},\\
&T_{c_1}\mP^6=T_{c_2}\mP^6=\{{\boldsymbol x}\in {\bR}^9:
x_4=x_8=0,bx_5-ax_7=0\}.
\end{eqnarray*}
Пусть $e^j$ $(j=1,\ldots 6)$ базис соответствующего пространства
$T_{c_k}\mP^6$:
\begin{equation}\label{basis}
\begin{array}{l}
e^1=\{1,0,0,0,0,0,0,0,0\}^t,\\
e^2=\{0,1,0,0,0,0,0,0,0\}^t,\\
e^3=\{0,0,1,0,0,0,0,0,0\}^t,\\
e^4=\{0,0,0,0,a,0,\mp b,0,0\}^t,\\
e^5=\{0,0,0,0,0,1,0,0,0\}^t,\\
e^6=\{0,0,0,0,0,0,0,0,1\}^t.
\end{array}
\end{equation}
Линеаризации векторных полей $\sgrad H, \sgrad K
$ и $\sgrad G$ порождают симплектические линейные операторы
\begin{eqnarray*}
A_H, A_K, A_G\,:T_{c_k}\mP^6\to T_{c_k}\mP^6.
\end{eqnarray*}
Пусть $\ad{x,y,z}$ означает матрицу
$$
\left(
\begin{matrix}
0&0&z\\
0&y&0\\
x&0&0\\
\end{matrix}
\right).
$$
Тогда матрицы операторов $A_H, A_K, A_G$ в выбранном базисе имеют
блочный вид
\begin{equation*}
\begin{array}{lll}
A_H=\left(
\begin{matrix}
\boldsymbol 0&A_1^{c_k}\\
A_2^{c_k}&\boldsymbol 0
\end{matrix}
\right),& A_G=\left(
\begin{matrix}
\boldsymbol 0&A_3^{c_k}\\
A_4^{c_k}&\boldsymbol 0
\end{matrix}
\right),& A_K=\left(
\begin{matrix}
\boldsymbol 0&A_5^{c_k}\\
A_6^{c_k}&\boldsymbol 0
\end{matrix}
\right),
\end{array}
\end{equation*}
где введены обозначения следующих матриц
\begin{equation*}
\begin{array}{ll}
A_1^{c_0}=A_1^{c_3}= \ad{a+b,\ds{-\frac{1}{2}},\ds{\frac{1}{2}}}, &
A_2^{c_0}=-A_2^{c_3}= \ad{-b,a,1},\\[4mm]
A_3^{c_0}=A_3^{c_3}= \ad{a b
(a+b),\ds{-\frac{b^2}{2}},\ds{\frac{a^2}{2}}},
& A_4^{c_0}=-A_4^{c_3}=ab\ad {-a,b,1},\\[4mm]
A_5^{c_0}=-A_5^{c_3}=(a-b) \ad{0,1,1}, &
A_6^{c_0}=A_6^{c_3}=-2(a-b)\ad{b,a,0},\\[4mm]
A_1^{c_1}=A_1^{c_2}= \ad{a-b,\ds{-\frac{1}{2}},\ds{\frac{1}{2}}} , &
A_2^{c_1}=-A_2^{c_2}=\ad{b,a,-1}, \\[4mm]
A_3^{c_1}=A_3^{c_2}= \ad{-a b
(a-b),\ds{-\frac{b^2}{2}},\ds{\frac{a^2}{2}}} , &
A_4^{c_1}=-A_4^{c_2}=ab\ad{a,b,1}
,\\[4mm]
A_5^{c_1}=-A_5^{c_2}=(a+b)\ad{0,1,1}, &
A_6^{c_1}=A_6^{c_2}=-2(a+b)\ad{-b,a,0}.
\end{array}
\end{equation*}
В силу (\ref{aneb}) подалгебра $\mA (c_k,{\mF })$, порожденная
операторами $A_H, A_G, A_K$, является картановской для всех
$k=0,\ldots,3$ и ее размерность равна $3$. В качества регулярного
элемента можно взять $A_H$. Характеристические уравнения имеют вид:
\begin{equation*}
\begin{array}{ll}
c_0: & (\mu^2+a+b)(2\mu^2+a)(2\mu^2+b)=0,\\
c_1: & (\mu^2+a-b)(2\mu^2+a)(2\mu^2-b)=0,\\
c_2: & (\mu^2-a+b)(2\mu^2-a)(2\mu^2+b)=0,\\
c_3: & (\mu^2-a-b)(2\mu^2-a)(2\mu^2-b)=0.
\end{array}
\end{equation*}
Соответствующие собственные значения различны и в каждой точке
разбиваются на три пары:
\begin{equation*}
\begin{array}{llll}
c_0: & \pm \ri \sqrt{a+b}; & \pm \ri  \sqrt{\ds{\frac{a}{2}}};  &
\pm
\ri \sqrt{\ds{\frac{b}{2}}}, \\
c_1: &\pm \ri \sqrt{a+b}; &  \pm \ri \sqrt{\ds{\frac{a}{2}}};  & \pm
\sqrt{\ds{\frac{b}{2}}},\\
c_2:&\pm \sqrt{a-b}; &  \pm \sqrt{\ds{\frac{a}{2}}};  &  \pm
\ri \sqrt{\ds{\frac{b}{2}}},\\
c_3:&\pm \sqrt{a+b}; &  \pm \sqrt{\ds{\frac{a}{2}}};  &  \pm
\sqrt{\ds{\frac{b}{2}}}.
\end{array}
\end{equation*}
Таким образом, точки $c_k$ -- невырожденные особенности и их тип
определяется тройкой целых чисел: $(3,0,0)$ для $c_0$;
$(2,1,0)$ для $c_1$; $(1,2,0)$ для $c_2$ и $(0,3,0)$ для
$c_3$.
\end{proof}

Заметим, что в точках $c_k$ встречаются три локальных критических
подсистемы -- подсистема $\mn$ и две части подсистемы $\mo$, которая
имеет в этих точках особенность типа самопересечения. В частности,
на плоскости $(s,\tau)$ каждая такая точка изображается двумя. В
этом смысле каждая особая точка $Q_{kj}$ диаграммы $\Sigma_3^*$
имеет свой тип (тип точки $c_k$ по отношению к некоторому выбранному
гладкому участку $\mo$ в ее окрестности). Соответствующее описание
критических точек ранга 0 в $\mP^6$, а также в  критических
подсистемах $\mn$ и $\mo$ сведено в табл.~\ref{tab41}.
{\renewcommand{\arraystretch}{1.5} \setlength{\extrarowheight}{-2pt}
\begin{table}[!htbp]
\centering
\begin{tabular}{|c|c|c|c|c|c|c|}
\multicolumn{7}{r}{\fts{Таблица
\myt\label{tab41}}}\\
\hline $\mK^0$ &
\begin{tabular}{c}
Тип в \\
 $\mP^6$
 \end{tabular}&
\begin{tabular}{c}
Образ в \\
 ${\bR}^3(h,k,g)$
 \end{tabular}&
  \begin{tabular}{c}
Образ в \\
 ${\bR}^2(m,h)$
 \end{tabular}  &
Тип в
 $\mn$
 & \begin{tabular}{c}
Образ в \\
 ${\bR}^2(s,h)$
 \end{tabular} & Тип в
 $\mo$ \\
\hline $c_0$ &$(3,0,0)$ &$P_0$ &$P_{01}$&центр-центр & \begin{tabular}{l}$Q_{01}$\\$Q_{02}$\end{tabular}&
\begin{tabular}{l}
центр-центр\\
центр-центр
\end{tabular}
\\
\hline $c_1$ &$(2,1,0)$ &$P_1$ &$P_{11}$&центр-седло &
\begin{tabular}{l}
$Q_{11}$\\
$Q_{12}$
\end{tabular}&
\begin{tabular}{l}
центр-седло\\
центр-центр
\end{tabular}
\\
\hline $c_2$ &$(1,2,0)$ &$P_2$ &$P_{21}$&центр-седло &
\begin{tabular}{l}
$Q_{21}$\\
$Q_{22}$
\end{tabular}&
\begin{tabular}{l}
седло-седло\\
центр-седло
\end{tabular}
\\
\hline $c_3$ &$(0,3,0)$ &$P_3$ &$P_{31}$&седло-седло &
\begin{tabular}{l}
$Q_{31}$\\
$Q_{32}$
\end{tabular}&
\begin{tabular}{l}
седло-седло\\
седло-седло
\end{tabular}
\\
\hline
\end{tabular}\,
\end{table}
}
\section{Невырожденные особенности ранга $1$}
Как было показано выше, невырожденные особенности ранга 1 состоят из
периодических решений семейств $\mL_k$ $(k=1,\ldots,9)$, образ
которых составляет одномерный комплекс
$\Sigma^1=\{\lambda_k,\delta_k\}$, в котором кривые разбиты на
участки точками комплекса $\Sigma^0$ (трансверсальное пересечение
одномерных кривых, бифуркации периодических решений при прохождении
через неподвижную точку), точками касания кривых между собой или с
кривыми $\Delta_i$ (вырожденные особенности).
\begin{theorem}
Точкам одномерного комплекса бифуркационной диаграммы $\Sigma$
соответствуют \textit{невырожденные\/} особенности $\{{\mL}_k\}$
ранга $1$ отображения момента
$$
{{\mF}=(H,K,G): \mP^6\to {\bR}^3},
$$
за исключением следующих значений энергии, при которых происходит
вырождение: на $\mL_2$ $h = 2a , \ds{\frac{3a^2+b^2}{2a}}$; на
$\mL_3$ $h = -2b $; на $\mL_4$ $h = 2b , \ds{\frac{a^2+3b^2}{2b}}$;
на $\mL_5$ $h =\pm 2\sqrt{a b}$. Вырожденным особенностям на
бифуркационных диаграммах отвечают точки $q_1 - q_3$, $e_1 - e_5$. В
зависимости от значений параметров тип невырожденных особенностей в
$\mP^6$ определяется таблицей~\ref{tab42}.
\end{theorem}

{\renewcommand{\arraystretch}{1.5} \setlength{\extrarowheight}{-2pt}
\begin{table}[!htbp]
\centering
\begin{tabular}{|c|l|l|}
\multicolumn{3}{r}{\fts{Таблица
\myt\label{tab42}}}\\
\hline $\mK^1$& \multicolumn{1}{c|}{Образ в
 ${\bR}^3(h,k,g)$}&\multicolumn{1}{c|}{Тип в
 $\mP^6$}\\
\hline ${\mL}_1$ &\begin{tabular}{l}
$\lambda_1=\lambda_{11}\cup\lambda_{12},$\\
$\lambda_{11}:-(a+b)<h<-(a-b),$\\
$\lambda_{12}:h>-(a-b),$
\end{tabular}
&\begin{tabular}{l} $(2,0,0)$
(центр-центр ранга 1)\end{tabular}\\

\hline ${\mL}_2$ &\begin{tabular}{l}
$\lambda_2=\lambda_{21}\cup\lambda_{22}\cup\lambda_{23},$\\
$\lambda_{21}:a-b<h<a+b$,\\
$\lambda_{22}:a+b<h<2a$,\, $\ds{h\ne\frac{3a^2+b^2}{2a}}$,\\
$\lambda_{23}: h>2a$
\end{tabular}&\begin{tabular}{l}
$\lambda_{21},\lambda_{22}: (0,2,0)$
(седло-седло ранга 1)\\
$\lambda_{23}: (1,1,0)$ (центр-седло ранга 1)
\end{tabular} \\

\hline ${\mL}_3$ &
\begin{tabular}{l}
$\lambda_3=\lambda_{31}\cup\lambda_{32}\cup\lambda_{33},$\\
$\lambda_{31}:-(a+b)<h<-2b$,\\
$\lambda_{32}:-2b<h<a-b$,\\
$\lambda_{33}: h>a-b$
\end{tabular}&\begin{tabular}{l}
$\lambda_{31}: (2,0,0)$ (центр-центр ранга 1)\\
$\lambda_{32},\lambda_{33}: (1,1,0)$ (центр-седло ранга 1)
\end{tabular}
\\

\hline ${\mL}_4$ &
\begin{tabular}{l}
$\lambda_4=\lambda_{41}\cup\lambda_{42}\cup\lambda_{43},$\\
$\lambda_{41}:-(a-b)<h<2b$,\\
$\lambda_{42}:2b<h<a+b$,\\
$\lambda_{43}: h>a+b$,\, $\ds{h\ne\frac{a^2+3b^2}{2b}}$
\end{tabular}&\begin{tabular}{l}
$\lambda_{41}:(1,1,0)$ (центр-седло ранга 1)\\
$\lambda_{42},\lambda_{43}: (0,2,0)$ (седло-седло ранга 1)
\end{tabular}\\

\hline ${\mL}_5$ &
\begin{tabular}{l}
$\lambda_5=\lambda_{50}\cup\left(\bigcup_{k=1}^4\,\lambda_{5k}\right),$\\
$\lambda_{51}:-(a+b)<h<-2\sqrt{ab}$\\
$\lambda_{50},-2\sqrt{ab}<h<2\sqrt{ab}$\\
$\lambda_{52}, \lambda_{53},\lambda_{54}:h>2\sqrt{ab}$
\end{tabular}&\begin{tabular}{l}
$\lambda_{51}:(2,0,0)$ (центр-центр ранга 1)\\
$\lambda_{50}:(0,0,1)$ (фокусная особенность ранга $1$)\\
$\lambda_{52}, \lambda_{53},\lambda_{54}: (0,2,0)$ (седло-седло ранга 1)
\end{tabular}
\\

\hline  ${\mL}_6$ & \begin{tabular}{l}$\lambda_6:h>-(a-b)$
\end{tabular}&\begin{tabular}{l} $(1,1,0)$
(центр-седло ранга 1)\end{tabular}\\

\hline ${\mL}_7$ & \begin{tabular}{l}$\delta_1:s\in(-b,0)$
\end{tabular}&\begin{tabular}{l} $(2,0,0)$
(центр-центр ранга 1)\end{tabular}\\

\hline ${\mL}_8$ & \begin{tabular}{l}$\delta_2:s\in(0,b)$
\end{tabular}&\begin{tabular}{l} $(1,1,0)$
(центр-седло ранга 1)\end{tabular}\\

\hline ${\mL}_9$ &\begin{tabular}{l}
$\delta_3=\delta_{31}\cup\delta_{32}$\\
$\delta_{31}: s\in(a,s_0)$\\
$\delta_{32}: s\in(s_0,+\infty)$
\end{tabular}
&\begin{tabular}{l} $\delta_{31}:(1,1,0)$
(центр-седло ранга 1)\\
$\delta_{32}: (2,0,0)$ (центр-центр ранга 1)
\end{tabular}\\
\hline
\end{tabular}\,
\end{table}
}
\begin{proof}
Рассмотрим произвольную точку $x_0\in{\mL}_k$. Определим следующие
функции
\begin{equation}\notag
\begin{array}{lll}
({\mL}_{1,2})& g_1=K-2(h\pm 2a)H, & g_2=G-a^2H,\\[4mm]
({\mL}_{3,4})& g_1=K-2(h\pm 2b)H, & g_2=G-b^2H,\\[4mm]
({\mL}_{5,6})& g_1=\pm abH-G, & g_2=K,\\[4mm]
({\mL}_{7,8,9})&g_1=2G-(p^2-\tau)H, & g_2=K.
\end{array}
\end{equation}
Положим также $g_3=H$. Поскольку все неподвижные точки имеют ранг 0
и уже исключены, то $dg_3(x_0) \ne 0$. Выбранные функции $g_1$,
$g_2$ в точке $x_0$ имеют особенность: $dg_1(x_0)=dg_2(x_0)=0$.
Линеаризации векторных полей $\sgrad g_k$ $(k=1,2)$ в точке $x_0$
дают линейные симплектические операторы $A_{g_k}: T_{x_0}\mP^6 \to
T_{x_0}\mP^6$ $(k=1,2)$. Непосредственно проверяется, что они
линейно независимы, то есть порождают подалгебру в $\sp(6,{\bR})$
размерности 2. Характеристическое уравнение для оператора $A_{g_1}$
имеет два нулевых корня: $\ker A_{g_1}=T_{x_0}\mL_k$. Остальная
часть характеристического многочлена имеет вид
\begin{equation}\label{char}
\begin{array}{ll}
{\mL}_{1,2}:&[\mu^2+4(a^2-b^2)(h\pm 2a)]\cdot[\mu^2\pm 8a(h\pm
2a)^2]=0,\\[4mm]
{\mL}_{3,4}:&[\mu^2-4(a^2-b^2)(h\pm 2b)]\cdot[\mu^2\pm 8b(h\pm
2b)^2]=0,\\[4mm]
{\mL}_{5,6}: & 4\mu^4\mp 2abh(a\mp
b)^2\mu^2\pm b^3a^3(a\mp b)^4=0,\\[4mm]
{\mL}_{7,8,9}: &\mu^4+u_{k}\mu^2+v_k=0 \qquad (k=1,2,3),
\end{array}
\end{equation}
где коэффициенты $u_k, v_k$ определяются по формулам
\begin{equation}\notag
\begin{array}{lcl}
u_{1,2}&=& - \ds{\frac{2}{s}}(\sqrt{\mathstrut {{a^2-s^2}}}\pm
\sqrt{\mathstrut {{b^2-s^2}}})^2 \Bigl[a^2b^2-s^4 \pm
6s^2\sqrt{\mathstrut{{(a^2-s^2)(b^2-s^2)}}}\Bigr],\\[4mm]
v_{1,2}&= &\pm
\ds{16\sqrt{\mathstrut{{(a^2-s^2)(b^2-s^2)}}}(\sqrt{\mathstrut
{{a^2-s^2}}}\pm \sqrt{\mathstrut {{b^2-s^2}}})^4
}\cdot V_{\pm}, \\[4mm]
u_3&=& \phantom{-} \ds{\frac{2}{s}}(\sqrt{\mathstrut
s^2-a^2}-\sqrt{\mathstrut s^2-b^2})^2 \Bigl[a^2b^2-s^4+
6s^2\sqrt{\mathstrut(s^2-a^2)(s^2-b^2)}\Bigr],\\[4mm]
v_3&=&\phantom{\pm}\ds{16\sqrt{\mathstrut(s^2-a^2)(s^2-b^2)}(\sqrt{\mathstrut
s^2-a^2}-\sqrt{\mathstrut s^2-b^2})^4
}\cdot V_{+},
\end{array}
\end{equation}
где
\begin{eqnarray*}
&&V_{\pm}=a^2b^2-s^4\pm 2s^2\sqrt{\mathstrut {{(a^2-s^2)(b^2-s^2)}}}.
\end{eqnarray*}

Нахождение многочленов в точках $\mL_1 - \mL_6$ труда не составляет.
Для множеств {$\mL_7 - \mL_9$}  при вычислении коэффициентов
$u_k,v_k$ характеристических многочленов использовалась запись
матрицы линеаризации соответствующего поля в базисе касательного
пространства к $\mP^6$, который получим из (\ref{basis}) заменой
последней тройки на векторы
\begin{equation*}
\begin{array}{l}
\{0,0,0, 0,-\alpha_3,\alpha_2, 0,-\beta_3,\beta_2\}^t,\\[4mm]
 \{0,0,0,\alpha_3,0,-\alpha_1, \beta_3,0,-\beta_1\}^t,\\[4mm]
 \{0,0,0,-\alpha_2,\alpha_1,0, -\beta_2,\beta_1,0\}^t
\end{array}
\end{equation*}
в подстановке выражений (\ref{equa:delone}).

При значениях параметров, указанных в табл.~\ref{tab42}, все корни
соответствующего характеристического уравнения различны и
разбиваются на пары, определяющие тип невырожденной особенности.

Для примера рассмотрим кривую $\delta_3$ ($s>a$). Характеристическое
уравнение в точках $\mL_9$ относительно $\mu^2$ имеет корни
\begin{equation*}
\begin{array}{l}
\mu^2_{(1)}= - 8s \sqrt{\mathstrut s^2-a^2}\sqrt{\mathstrut s^2-b^2}
\bigl(\sqrt{\mathstrut s^2-a^2}-\sqrt{\mathstrut s^2-b^2}\bigr)^2 <
0,\\
\mu^2_{(2)}=\ds{\frac{2}{s}}\bigl(\sqrt{\mathstrut
s^2-a^2}-\sqrt{\mathstrut s^2-b^2}\bigr)^2\bigl(s^4-a^2 b^2-2
s^2\sqrt{\mathstrut s^2-a^2}\sqrt{\mathstrut s^2-b^2}\bigr).
\end{array}
\end{equation*}
Последнее выражение меняет знак при переходе через значение $s_0$
согласно уравнению (\ref{equa:s0}): $\mu^2_{(2)}>0$ при $s\in
(a,s_0)$ и $\mu^2_{(2)}< 0$ при $s > s_0$. В первом случае получаем
особенность <<центр-седло>>, во втором -- <<центр-центр>>. По
отношению к критической подсистеме $\mo$ все точки периодических
решений $\mL_7 - \mL_9$ имеют эллиптический тип, так как
соответствующие кривые являются внешней границей области
существования движений. Поэтому в критической подсистеме $\mm$
тип рассматриваемых точек гиперболический при $s\in (a,s_0)$ и
эллиптический при $s > s_0$. Теперь эти утверждения относительно
обеих подсистем доказаны аналитически.

Отметим также, что согласно табл.~\ref{tab42}
получено аналитическое доказательство
существования  невырожденной фокусной
особенности ранга $1$ на части множества
$\mL_5$, которая отображается в $\lambda_{50}$, то есть при значениях энергии
${-2\sqrt{ab}<h<2\sqrt{ab}}$. В
изоэнергетических сечениях $\Sigma(h)$ имеем
изолированную точку на бифуркационной
диаграмме.

Все остальные случаи, отраженные в таблице, также вытекают из
анализа корней характеристических многочленов~(\ref{char}).
Вырождения особенностей соответствуют наличию кратного корня.
\end{proof}

Теперь можно \textit{однозначно} установить
топологию слоения Лиувилля в окрестности
невырожденных положений равновесия
\textit{внутри} критических подсистем $\mn$ и
$\mo$ (напомним, что в $\mm$ они не попадают) в
виде прямого или почти прямого произведения
2-атомов, если воспользоваться, например,
методом круговых молекул \cite{BolFom}, \cite{BolFomRich}. Однако для
его применения необходима информация о
количестве двумерных торов в каждой области,
регулярной для подсистемы, и знание типа
невырожденных особенностей ранга $0$ и $1$.
Количество торов, равно как и количество
периодических решений на каждом совместном
уровне выбранной пары частных интегралов,
установлено в \cite{KhSav}, \cite{Kh2009}. Тип
особенностей установлен выше. Для примера,
окрестность особого слоя невырожденной
особенности $c_3$ \textit{внутри} $\mn$ можно
представить в виде прямого произведения
$B{\times}B$ двух атомов типа $B$,
\textit{внутри} $\mo$ на одном гладком
четырехмерном листе критического многообразия
-- $B{\times B}$, а на другом -- почти прямое
произведение $(B{\times}C_2)/\bZ_2$, где группа
${\bZ}_2$ действует на каждом из сомножителей
как центральная симметрия. Этого достаточно для
того, чтобы понять, как устроена окрестность
${\cal U}(\LS)$ особого слоя ${\LS}$ в $\mP^6$,
содержащего особую точку $c_3$ типа
<<седло-седло-седло>> ранга 0. Нужно взять
модельную особенность типа прямого произведения
трех 2-атомов $B{\times}B{\times}C_2$ и
рассмотреть на нем покомпонентное действие
образующей $e$ группы ${\bZ}_2$ по правилу
$e(B{\times}B{\times}C_2)=(\beta(B){\times}\mathop{\rm
id}\nolimits (B){\times}\alpha(C_2))$, где
$\alpha$ и $\beta$ -- центральные симметрии
симметричных атомов $B$ и $C_2$. После
факторизации окрестность ${\cal U}(\LS)$
получается в виде почти прямого произведения
$(B{\times}B{\times}C_2)/{\bZ}_2$. Таким
образом, из 32 разных особенностей типа
<<седло-седло-седло ранга 0>>, описанных в
 \cite{Kal}, \cite{Osh1}, \cite{Osh2}, для волчка
Ковалевской в двойном поле сил реализуется
только одна, указанная выше. В
таблице~\ref{tab43} содержится информация об
окрестностях особых слоев в $\mn$, $\mo$ и
$\mP^6$, содержащих неподвижные точки системы.
{\renewcommand{\arraystretch}{1.5} \setlength{\extrarowheight}{-2pt}
\begin{table}[!htbp]
\centering
\begin{tabular}{|c|c|c|c|c|c|c|}
\multicolumn{7}{r}{\fts{Таблица
\myt\label{tab43}}}\\
\hline $\mK^0$ &
\begin{tabular}{c}
Образ в \\
 ${\bR}^2(m,h)$
 \end{tabular}&
 \begin{tabular}{c}
${\cal U}(\LS)$\\ в
 $\mn$
 \end{tabular}&
  \begin{tabular}{c}
Образ в \\
 ${\bR}^2(s,h)$
 \end{tabular}  &
 \begin{tabular}{c}
${\cal U}(\LS)$\\ в
 $\mo$
 \end{tabular}
 & \begin{tabular}{c}
Образ в \\
 ${\bR}^3(h,k,g)$
 \end{tabular} & \begin{tabular}{c}
${\cal U}(\LS)$\\ в
 $\mP^6$
 \end{tabular} \\
\hline $c_0$ &$P_{01}$ &$A{\times} A$ &\begin{tabular}{c}
$Q_{01}$\\
$Q_{02}$
\end{tabular}&\begin{tabular}{c}
$A{\times} A$\\
$A{\times} A$
\end{tabular}&
$P_0$ & $A{\times} A{\times} A$
\\
\hline $c_1$ &$P_{11}$ &$A{\times} B$ &\begin{tabular}{c}
$Q_{11}$\\
$Q_{12}$
\end{tabular}&\begin{tabular}{c}
$A{\times} B$\\
$A{\times} A$
\end{tabular}&
$P_1$& $A{\times} A{\times} B$
\\
\hline $c_2$ &$P_{21}$ &$A{\times} B$ &\begin{tabular}{c}
$Q_{21}$\\
$Q_{22}$
\end{tabular}&\begin{tabular}{c}
$(B{\times} C_2)/{\bZ}_2$\\
$A{\times} B$
\end{tabular}&
$P_2$& $(A{\times} B{\times} C_2)/{\bZ}_2$
\\
\hline $c_3$ &$P_{31}$ &$B{\times} B$ &\begin{tabular}{c}
$Q_{31}$\\
$Q_{32}$
\end{tabular}&\begin{tabular}{c}
$(B{\times} C_2)/{\bZ}_2$\\
$B{\times} B$
\end{tabular}
& $P_3$&$(B{\times} B{\times} C_2)/{\bZ}_2$
\\
\hline
\end{tabular}\,
\end{table}
}

\section{Невырожденные особенности ранга $2$}\label{s5}
Множество $\mK^2$ всех критических точек ранга $2$ есть объединение
трех систем $\mm$, $\mn$ и $\mo$ за вычетом уже исследованных точек
множества $\mK^0 \cup \mK^1$.

Отметим сразу, что для системы $\mm$, заданной
уравнениями (\ref{eq1.4}), (\ref{eq1.7}),
каждый регулярный тор ${\mathbb{T}}^2\in \{x\in
\mm : H=h,F=f\}$ является
\textit{эллиптическим}. Это вытекает из того,
что интеграл $K=Z_1^2+Z_2^2$ есть положительная
всюду функция и обращается в ноль на $\mm$.

\begin{theorem} Все критические точки ранга 2 на многообразии $\mm$
являются невырожденными типа $(1,0,0)$ за исключением точек
нулевого уровня интеграла $F$.
\end{theorem}
\begin{proof}
Напомним, что на $\mm$ нет критических точек интеграла $H$, и
отметим, что кроме множеств $\mL_7 - \mL_9$ на $\mm$ нет точек
зависимости интегралов $H$ и $G$. Последнее доказано в
\cite{Kh2005}. Таким образом, эти два интеграла регулярны и
независимы на $\mm \cap \mK^2$. В то же время всюду на $\mm$ имеем
$dK=0$. Характеристическое уравнение оператора $A_K$ в $\bR^9$ при
условии (\ref{eq1.7}) легко выписывается, имеет семь нулевых корней,
а оставшийся сомножитель $\mu^2+4 f^2$ имеет два различных мнимых
корня  при $f\ne 0$, что и доказывает теорему.
\end{proof}

В образе получаем множество $\Delta_1 \subset \bR^2(h,k,g)$, для
которого тем самым получено обоснование вырожденности
соответствующих критических точек.

На многообразии $\mn$, заданном уравнениями (\ref{eq1.4}),
(\ref{eq1.8}), система (\ref{eq1.1}) имеет явное алгебраическое
решение \cite{KhSav}:
\begin{equation}\label{ne4_18}
\begin{array}{l}
\displaystyle{\alpha _1  = \frac{\mathstrut 1} {{2(s_1  - s_2 )^2
}}[(s_1 s_2
- a^2 )\psi + S_1 S_2 \varphi _1 \varphi _2 ], }\\
\displaystyle{\alpha _2  = \frac{\mathstrut 1} {{2(s_1  - s_2 )^2
}}[(s_1 s_2
- a^2)\varphi _1 \varphi _2  - S_1 S_2 \psi ], }\\
\displaystyle{\beta _1  =  - \frac{\mathstrut 1} {{2(s_1  - s_2 )^2
}}[(s_1
s_2  - b^2)\varphi _1 \varphi _2  - S_1 S_2 \psi ], }\\
\displaystyle{\beta _2  = \frac{\mathstrut 1} {{2(s_1  - s_2 )^2
}}[(s_1 s_2
- b^2 )\psi + S_1 S_2 \varphi _1 \varphi _2 ], }\\
\displaystyle{\alpha _3  = \frac{\mathstrut r} {{s_1  - s_2 }}S_1
,\quad
\beta _3  =\frac{r} {{s_1  - s_2 }}S_2, }\\
\displaystyle{\omega _1  = \frac{\mathstrut r} {{2(s_1  - s_2
)}}(\ell - 2ms_1 )\varphi _2,\quad \omega _2  = \frac{r} {{2(s_1 -
s_2
)}}(\ell  - 2ms_2)\varphi _1, }\\
\displaystyle{\omega _3  = \frac{\mathstrut 1} {{s_1  - s_2 }}(S_2
\varphi _1 - S_1 \varphi _2 )}.
\end{array}
\end{equation}
Здесь обозначено
\begin{equation}\label{ne4_19}
\begin{array}{l}
\displaystyle{\psi = 4ms_1 s_2  - 2\ell (s_1  + s_2 ) +
\frac{1}{m}(\ell ^2  - 1)},  \\
\displaystyle{S_1 = \sqrt {s_1^2 - a^2},}\quad
\displaystyle{\varphi_1 = \sqrt {\mstrut - \varphi (s_1)},}\quad
\displaystyle{S_2 = \sqrt {b^2 - s_2^2},}\quad
\displaystyle{\varphi_2 = \sqrt {\mstrut\varphi (s_2)},}\\
\displaystyle{\varphi (s) = 4ms^2  - 4\ell s + \frac{1} {m}(\ell ^2
- 1)},
\end{array}
\end{equation}
и постоянная $\ell$ связана с $h,m$ уравнением (\ref{lhm}).
Зависимость $s_1 ,s_2 $ от времени задана уравнениями
\begin{equation}\label{dif2}
\frac{ds_1}{dt}=\frac{1}{2}\sqrt{(a^2-s_1^2)\mstrut\varphi
(s_1)},\quad
\frac{ds_2}{dt}=\frac{1}{2}\sqrt{(b^2-s_2^2)\mstrut\varphi (s_2)}.
\end{equation}

\begin{theorem} Все критические точки ранга 2 на многообразии $\mn$ невырождены,
за исключением точек, лежащих в прообразе кривых $\Delta_1,
\Delta_2$. При этом они имеют тип $(1,0,0)$ для $m>0$ и
$(0,1,0)$ для $m<0$.
\end{theorem}
\begin{proof}
В качестве единственного интеграла, имеющего особенность в каждой
точке $\mn \cap \mK^2$, благодаря наличию достаточно простого
уравнения (\ref{eqpov2}) листа $\pov_2$, удобно взять
соответствующую функцию (\ref{eqpov}):
\begin{equation}\label{fi2}
\Phi = \phi_2(H,K,G)=(2G-p^2H)^2-r^4 K.
\end{equation}
Для нее характеристическое уравнение оператора $A_\Phi$ после
необходимой факторизации по нулевому корневому подпространству в
подстановке явных выражений (\ref{ne4_18}) примет вид
\begin{equation*}\label{eq5.3}
\mu^2+4r^{12} \ell^2 m=0,
\end{equation*}
и за исключением случаев $m=0$ ($\Delta_1$) и $\ell=0$ ($\Delta_2$)
имеет два различных корня, чисто мнимых при $m>0$ и вещественных при
$m<0$. Теорема доказана.
\end{proof}

\begin{remark}
Как показано в \cite{KhSav}, $\mn$ состоит из
критических точек нулевого уровня функций $2G
-p^2H \pm r^2 \sqrt{K}$, одна из которых
является суммой квадратов двух гладких
регулярных функций ($m>0$), а другая --
разностью квадратов ($m<0, \ell \ne 0$). Отсюда
следует эллиптичность двумерных торов в первом
случае, и гиперболичность во втором. Здесь мы
вдобавок строго доказали невырожденность таких
точек. Интересно отметить связь вырожденности
критических точек с аналитическим решением. При
$m \to 0$ многочлен $\varphi$ в (\ref{ne4_19})
имеет предел $\varphi(s) \to 2h -4 s$, поэтому
при $m=0$ падает степень подкоренного выражения
в дифференциальных уравнениях (\ref{dif2}). При
$\ell=0$ переменные $s_1,s_2$ входят в правые
части этих уравнений только в четных степенях,
в связи с чем возникает дополнительная
симметрия.
\end{remark}

На многообразии $\mo$, заданном уравнениями
(\ref{eq1.4}), (\ref{eq1.10}), явное
алгебраическое решение системы (\ref{eq1.1})
указано в \cite{Kh2009,Kh2007}. Пусть $s,\tau$,
как и выше, -- постоянные интегралов
(\ref{eq1.11}), (\ref{intT}), последняя связана
с $h$ линейной зависимостью (\ref{ne1_38}).
Введем также обозначения
\begin{equation*}\label{ne4_38}
\begin{array}{l}
\sigma = \tau^2-2p^2 \tau+r^4, \quad
\chi=\sqrt{\ds{\frac{4s^2\tau+\sigma}{4s^2}}}, \quad
\varkappa=\sqrt{\sigma}.
\end{array}
\end{equation*}
Тогда на любом совместном уровне интегралов внутри
$\mo$ имеем
\begin{eqnarray}
& & \begin{array}{l} \displaystyle{\alpha_1=\frac{(\mathcal{A}-r^2
U_1 U_2)(4 s^2 \tau+U_1 U_2)-(\tau+r^2) M_1 N_1 M_2 N_2 V_1 V_2}{4
r^2 s\, \tau (U_1+U_2)^2},
} \\[3mm]
\displaystyle{\alpha_2=\ri \frac{(\mathcal{A} -r^2 U_1 U_2)V_1 V_2
-(4 s^2 \tau+U_1 U_2)(\tau+r^2) M_1 N_1 M_2 N_2}{4 r^2 s\, \tau
(U_1+U_2)^2},
}\\[3mm]
\displaystyle{\beta_1=\ri \frac{(\mathcal{B}+r^2 U_1 U_2)V_1 V_2-(4
s^2 \tau+U_1 U_2)(\tau-r^2) M_1 N_1 M_2 N_2}{4 r^2 s\, \tau
(U_1+U_2)^2},
} \\[3mm]
\displaystyle{\beta_2=-\frac{(\mathcal{B} + r^2 U_1 U_2)(4 s^2 \tau
+ U_1 U_2)-(\tau-r^2) M_1 N_1 M_2 N_2 V_1 V_2}{4 r^2 s\, \tau
(U_1+U_2)^2},
}\\[3mm]
\displaystyle{\alpha_3= \frac{R }{r \sqrt{\mathstrut 2} } \, \frac
{M_1 M_2}{t_1+t_2},}\quad \displaystyle{\beta_3= - \ri \frac{R}{r
\sqrt{\mathstrut 2} } \, \frac { N_1 N_2}{t_1+t_2},}\label{ne4_39}
\end{array}\\[3mm]
 & & \nonumber\begin{array}{l}\displaystyle{\omega_1=  \frac{R }{4 r
s\, \sqrt{s \,\tau}} \, \frac{M_2 N_1
U_1 V_2 + M_1 N_2 U_2 V_1}{ t_1^2-t_2^2},} \\[3mm]
\displaystyle{\omega_2= -\frac{\ri\,R}{4 r s\, \sqrt{s \,\tau} } \,
\frac{M_2 N_1 U_2 V_1 + M_1 N_2 U_1 V_2}{t_1^2-t_2^2},} \\[3mm]
\displaystyle{\omega_3= \frac{U_1-U_2}{\sqrt{2s\tau}}\frac{M_2 N_2
V_1 - M_1 N_1 V_2}{t_1^2-t_2^2}.}
\end{array}
\end{eqnarray}
Здесь
\begin{equation*}\label{ne4_41}
\begin{array}{ll}
K_1=\sqrt{\mstrut t_1+\varkappa}\,, & K_2=\sqrt{\mstrut t_2+\varkappa}\,,\\
L_1=\sqrt{\mstrut t_1-\varkappa}\,, & L_2=\sqrt{\mstrut t_2-\varkappa}\,,\\
M_1=\sqrt{\mstrut t_1+\tau+r^2}\,, & M_2=\sqrt{\mstrut
t_2+\tau+r^2}\,,
\\[1.5mm]
N_1=\sqrt{\mstrut t_1+\tau-r^2}\,, &
N_2=\sqrt{\mstrut t_2+\tau-r^2}\,,\\
V_1 = \sqrt{\mstrut 4s^2\chi^2-t_1^2}\,, & V_2 = \sqrt{\mstrut
4s^2\chi^2-t_2^2},\\
U_1=\sqrt{\mstrut t_1^2-\sigma}=K_1 L_1, & U_2=\sqrt{\mstrut
t_2^2-\sigma}=K_2 L_2,\\
R=\ds{\frac{1}{\sqrt{\mathstrut 2}}}(K_1 K_2 + L_1 L_2),\\
\end{array}
\end{equation*}
и
\begin{equation}\notag
\begin{array}{l}
\mathcal{A}=[(t_1+\tau+r^2)(t_2+\tau+r^2)-2(p^2+r^2)r^2]\tau,\\
\mathcal{B}=[(t_1+\tau-r^2)(t_2+\tau-r^2)+2(p^2-r^2)r^2]\tau.
\end{array}
\end{equation}
Зависимость вещественных переменных $t_1,t_2$ от времени описывается
уравнениями
\begin{equation*}\label{dif3}
\begin{array}{l}
\ds{(t_1-t_2)\frac{dt_1}{dt}=\sqrt{\frac{1}{2s\tau}(4s^2\chi^2-t_1^2)(t_1^2-\sigma)[r^4-(t_1+\tau)^2]}},\\
\ds{(t_1-t_2)\frac{dt_2}{dt}=\sqrt{\frac{1}{2s\tau}(4s^2\chi^2-t_2^2)(t_2^2-\sigma)[r^4-(t_2+\tau)^2]}}.
\end{array}
\end{equation*}

\begin{theorem}
Все регулярные двумерные торы, заданные формулами $(\ref{ne4_39})$,
состоят из невырожденных критических точек ранга 2 отображения
момента ${\mF}$, за исключением значений параметров, отвечающих
множествам $\Delta_2, \Delta_3$. Тор является
эллиптическим, если значение $\tau s
[s^4-(a^2+b^2-\tau)s^2+a^2 b^2 ]$ отрицательно, и
гиперболическим, если оно положительно.
\end{theorem}
\begin{proof}
В данном случае в качестве интеграла, имеющего особенность на $\mo$
также можно взять функцию, аналогичную (\ref{fi2}), полученную,
например, исключением $s$ из уравнений (\ref{ghar2}) поверхности
$\pov_3$. Однако результат получается слишком громоздким, и такой
подход нерационален. Здесь удобно рассмотреть функцию с
неопределенными множителями Лагранжа, которая введена в
\cite{Kh2005} для вывода уравнений критических подсистем,
\begin{equation*}
\Psi = 2 G- (p^2-\tau) H+s K.
\end{equation*}
Как показано в \cite{Kh2005}, {\it после} записи условия наличия
критической точки у функции $\Psi$ константы $s,\tau$ на $\mo$
оказываются значениям частных интегралов $S,T$. Поэтому и при
вычислении характеристического многочлена оператора $A_\Psi$ считаем
$s,\tau$ константами, а затем подставляем в найденное выражение
значения (\ref{ne4_39}). Получим характеристическое уравнение в виде
\begin{eqnarray*}\label{char:O4}
&&\mu^2-\frac{2\tau}{s}\Bigl[s^4-(p^2-\tau)s^2+\frac{p^4-r^4}{4}\Bigr]=0.
\end{eqnarray*}
Отсюда следует утверждение теоремы.
\end{proof}

\section{Заключение}
В работе выполнена полная классификация особых точек отображения
момента неприводимой интегрируемой гамильтоновой системы с тремя
степенями свободы -- задачи о движении волчка типа Ковалевской в
двойном силовом поле, интегрируемость которой установлена
А.Г.\,Рейманом и М.А.\,Семеновым\-Тян-Шанским. Предъявлены явные
формулы характеристических уравнений для собственных чисел
соответствующих симплектических операторов, которые и определяют тип
невырожденной особенности. Для вывода характеристических уравнений
используется параметризация периодических решений и двумерных торов
в критических подсистемах, полученная в \cite{KhSav}, \cite{Kh2009}
как составная часть разделения переменных и построения
алгебраического решения. Полученные результаты составляют
аналитическую основу для описания топологии особых слоев слоения
Лиувилля рассматриваемой системы.
\end{fulltext}


\end{document}